\documentclass[11pt,a4paper]{article}

\newcommand{\conf}[1]{}
\newcommand{\full}[1]{#1}

\usepackage{booktabs}

\usepackage{amsmath}
\usepackage{amssymb}

\usepackage{amsthm}
\usepackage{latexsym}
\usepackage{mathtools}

 \theoremstyle{plain}
 \newtheorem{theorem}{Theorem}[section]
 \newtheorem{corollary}[theorem]{Corollary}
 \newtheorem{proposition}[theorem]{Proposition}
 
 \newtheorem{lemma}[theorem]{Lemma}

 \theoremstyle{definition}
 \newtheorem{remark}[theorem]{Remark}

\newcommand{\ZZ}{{\mathbb{Z}}}
\newcommand{\RR}{{\mathbb{R}}}
\DeclareMathOperator{\preclex}{\prec_{lex}}
\DeclareMathOperator{\preceqlex}{\preceq_{lex}}
\DeclareMathOperator{\tol}{tol}
\newcommand{\remove}[1]{}

\usepackage{geometry}
\usepackage{fullpage}

\title{A Unified Model of Congestion Games with Priorities: 
Two-Sided Markets with Ties, 
Finite and Non-Affine Delay Functions, 
and 
Pure Nash Equilibria}

 \author{
 Kenjiro Takazawa\thanks{Department of Industrial and Systems Engineering, Faculty of Science and Engineering, 
 Hosei University, Tokyo 184-8584, Japan.  
 \texttt{takazawa@hosei.ac.jp}.  
 Supported by 
 JSPS KAKENHI Grant Numbers 
 JP20K11699, 24K02901, JP24K14828, 
 Japan.
 }}

\remove{
}

\remove{
\EventEditors{John Q. Open and Joan R. Access}
\EventNoEds{2}
\EventLongTitle{42nd Conference on Very Important Topics (CVIT 2016)}
\EventShortTitle{CVIT 2016}
\EventAcronym{CVIT}
\EventYear{2016}
\EventDate{December 24--27, 2016}
\EventLocation{Little Whinging, United Kingdom}
\EventLogo{}
\SeriesVolume{42}
\ArticleNo{23}
}

\date{July 2024}
\begin{document}

\maketitle
\begin{abstract}
The study of equilibrium concepts in congestion games and two-sided markets with ties has been a primary topic in game theory, economics, and computer science.  Ackermann, Goldberg, Mirrokni, R\"oglin, V\"ocking (2008) gave a common generalization of these two models, in which a player more prioritized by a resource produces an infinite delay on less prioritized players.  While presenting several theorems on pure Nash equilibria in this model, Ackermann et al.\ posed an open problem of how to design a model in which more prioritized players produce a large but finite delay on less prioritized players.  In this paper, we present a positive solution to this open problem by combining the model of Ackermann et al.\ with a generalized model of congestion games due to Bil\`o and Vinci (2023).  In the model of Bil\`o and Vinci, the more prioritized players produce a finite delay on the less prioritized players, while the delay functions are of a specific kind of affine function, and all resources have the same priorities. By unifying these two models, we achieve a model in which the delay functions may be finite and non-affine, and the priorities of the resources may be distinct.  We prove some positive results on the existence and computability of pure Nash equilibria in our model, which extend those for the previous models and  support the validity of our model. 
\end{abstract}

\maketitle

% Paper body
\section{Introduction}
\label{SECintro}

The study of equilibrium concepts in 
noncooperative games 
is a primary topic in the fields of game theory, economics, and computer science. 
In particular, 
the models of \emph{congestion games} and \emph{two-sided markets with ties} 
have played important roles in the literature. 

\emph{Congestion games}, introduced by 
Rosenthal \cite{Ros73a} in 1973, 
represent the human behaviour of avoiding congestion.  
Each \emph{player} chooses a strategy, 
which is a set of \emph{resources}. 
If a resource is shared by many players, 
then much delay is imposed on those players. 
The objective of a player is to minimize the total delay of the resources in her strategy. 
Rosenthal \cite{Ros73a} proved that 
every congestion game is a \emph{potential game}. 
A noncooperative game is called a potential game if it admits a \emph{potential function}, 
the existence of which guarantees 
the existence of a pure Nash equilibrium. 
Moreover, 
Monderer and Shapley \cite{MS96} proved the converse: 
every potential game can be represented as a congestion game. 
%%See Section \ref{SECconge} for a formal description. 
On the basis of these results, 
congestion games 
are recognized as a fundamental model 
in the study of pure Nash equilibria in noncooperative games  (see, e.g.,\ \cite{BV23book,Rou16}). 

A \emph{two-sided market} consists of \emph{agents} and \emph{markets}, 
which have preferences over the other side. 
Each agent chooses a set of markets. 
On the basis of the choices of the agents, 
the markets determine an assignment of the players to the markets 
according to their preferences over the agents. 
The objective of a player is to maximize her payoff, 
which is determined by the assignment. 
Typical special cases of a two-sided market are the stable matching problem  and the 
Hospitals/Residents problem. 
Since the pioneering work of Gale and Shapley \cite{GS62}, 
analyses on equilibria have been a primary topic in the study of two-sided markets, 
and 
a large number of generalized models have been proposed. 
In particular, 
a typical generalization of allowing \emph{ties} in the preferences \cite{Irv94} 
critically changes the difficulty of the analyses (see \cite{GI89,Man13}), 
and attracts intensive interests  \cite{GMMY22,HMY19,HIMY15,HK15,Kam15,Kam19,Kav22,MS10,Yok21}.

In 2008, 
Ackermann, Goldberg, Mirrokni, R\"oglin, and V\"oking \cite{AGMRV08}
introduced a model 
which commonly generalizes congestion games and two-sided markets with ties, 
and is referred to as \emph{congestion games with priorities}. 
\full{This model is briefly described as follows. }%
Each resource $e$ has priorities (preferences) with ties over the players. 
Among the players choosing $e$ in their strategies, 
only the players most prioritized by $e$ receive a finite delay from $e$, 
and the other players receive an infinite delay. 
In other words, 
only the most prioritized players are accepted. 
It is clear that this model generalizes congestion games, 
and it also generalizes a certain model of two-sided markets with ties, 
\emph{correlated two-sided markets with ties}. 

For several classes of their model, 
Ackermann et al.\ \cite{AGMRV08} presented some positive results on 
the existence and computability of pure Nash equilibria. 
These results are summarized in Table \ref{TABAck} 
and 
will be formally described 
in Section \ref{SECAGMRV}. 
In \emph{player-specific congestion games}, 
each resource $e$ has a specific delay function 
%%$d_{i,e}\colon \ZZ_{++} \to \RR_+$ 
$d_{i,e}$ 
for each player $i$.  
In a \emph{singleton congestion game}, 
every strategy of every player consists of a single resource. 
In a \emph{matroid congestion game}, 
the strategies of each player are the bases of a \emph{matroid}. 
% While presenting several results on pure Nash equilibria in their model, 
% Ackermann et al.\ \cite{AGMRV08} 
% posed an open problem on how to design a model in which more prioritized players produce a finite delay on 
% less prioritized players. 
\begin{table}
\centering
\caption{Results of Ackermann et al.\ \cite{AGMRV08}. 
``NPS'' stands for ``Non-Player-Specific,''
while ``PS'' stands for ``Player-Specific.''
%%``BR'' stands for ``Better Responses,''  and 
``Polynomial BR Dynamics'' means that 
there exists a sequence of a polynomial number of best responses reaching a pure Nash equilibrium. 
% ``Potential Function'' means that 
% there exists a potential function. 
% ``Algorithm'' means an algorithm for computing a pure Nash equilibrium.
}
\label{TABAck}
\small
\begin{tabular}{lll}
\toprule
						&\textbf{Consistent Priorities} 	     & \textbf{Inconsistent Priorities} \\ \midrule
%%\textbf{Non-Player-Specific} 	& \begin{tabular}{l} 
\textbf{NPS} 	& \begin{tabular}{l} 
                            \textbf{Polynomial BR Dynamics} \\ 
                            - Singleton Game (Theorem \ref{THMsingletonCP})  \\ 
                            - Matroid Game (Theorem \ref{THMmatroidCP})
                            \end{tabular}
                                                                    & \begin{tabular}{l}
                                                                    \textbf{Potential Function} \\
                                                                    - Singleton Game (Theorem \ref{THMsingletonIP})\\
                                                                    - Matroid Game (Theorem \ref{THMmatroidIP})
                                                                    \end{tabular}\\ \midrule
%%\textbf{Player-Specific}      	&   ---                                     &  \begin{tabular}{l}
\textbf{PS}      	&   ---                                     &  \begin{tabular}{l}
                                                                    \textbf{Potential Function} \\
                                                                    - Two-Sided Singleton Market (Theorem \ref{THMcorrelated}) \\
                                                                    - Two-Sided Matroid Market (Theorem \ref{THMcorrelatedmatroid}) \\
                                                                       \textbf{Polynomial Algorithm}   \\
                                                                       - Singleton Game (Theorem \ref{THMpssingletonIP})\\
                                                                       - Matroid Game (Theorem \ref{THMpsmatroidIP})
                                                                    \end{tabular}
\\ \bottomrule
\end{tabular}
\end{table}

Meanwhile, 
Ackermann et al.\ \cite{AGMRV08} posed an open question of 
how to design a model in which 
the less prioritized players receive a finite delay
caused by 
the more prioritized players. 
We appreciate the importance of this question because 
such a model can include the many-to-many models of the stable matching problem \cite{BB00,BAM07,EMMM12,Fle03,Sot99a,Sot99b} (see also \cite{Man13}). 
In congestion games, 
the fact that 
each resource accepts multiple players is essential, 
since the number of those players determine the cost of the resource. 
In the models of stables matchings in which each market accepts multiple agents, 
the preferences of the markets indeed affect the stability of the matchings, 
but 
it is not the case that only the most preferred agents are accepted. 
Thus, 
such a generalization of congestion games with priorities suggested in \cite{AGMRV08} 
is crucial to attain a more reasonable generalization of the stable matching problem. 

The contributions of the paper are described as follows. 
We first point out that 
a generalized model of congestion games by Bil\`o and Vinci \cite{BV23tcs}
partially answers to the question posed by Ackermann et al.\ \cite{AGMRV08}. 
In their model, 
the players more prioritized by a resource indeed produce a finite delay on the less prioritized players. 
Meanwhile, 
this model only covers a special case of the model of Ackermann et al.\ \cite{AGMRV08} 
in which 
the delay functions are of a specific kind of affine functions 
and all resources have the same priorities over the players. 
We refer to the model of 
%%Bil\`o and Vinci 
\cite{BV23tcs} as 
a \emph{priority-based affine congestion game with consistent priorities}.

A main contribution of this paper is to design a model which gives a positive and full answer to the open problem of 
Ackermann et al.\ \cite{AGMRV08}. 
By 
unifying the models of 
Ackermann et al.\ \cite{AGMRV08} 
and 
Bil\`o and Vinci \cite{BV23tcs}, 
we present a model of congestion games with priorities 
in which the more prioritized players produce a finite delay on 
the less prioritized players, 
the delay function may be non-affine, 
and 
the priorities of the resources may be inconsistent. 
We refer to our model as a \emph{priority-based congestion games with (in)consistent priprities}. 
We then prove some positive results on 
the existence and computability of pure Nash equilibria 
in our model, 
which extend those for the previous models \cite{AGMRV08,BV23tcs} 
and support the validity of our model. 
Our technical results are summarized in Table \ref{TABours}.

\begin{table}
\centering
\caption{Summary of Our Results. 
``NPS'' stands for ``Non-Player-Specific,''
while ``PS'' stands for ``Player-Specific.''
``Polynomial BR Dynamics'' means that 
there exists a sequence of a polynomial number of better responses reaching a pure Nash equilibrium. 
``PNE'' stands for ``Pure Nash Equilibrium.''
%%``Algorithm'' means an algorithm for computing a pure Nash equilibrium.
}
\label{TABours}
\small
\begin{tabular}{lll}
\toprule
						&\textbf{Consistent Priorities} 	     & \textbf{Inconsistent Priorities} \\ \midrule
%%\textbf{Non-Player-Specific} 	& \begin{tabular}{l} 
\textbf{NPS} 	& \begin{tabular}{l} 
                            \textbf{Polynomial BR Dynamics} \\ 
                            - Singleton Game (Theorem \ref{THMpbsingletonCP}) \\ 
                            - Matroid Game (Theorem \ref{THMpspbmatroidCP}) \\
%%                            \textbf{Existence of a Pure Nash Equilibrium} \\
                            \textbf{Existence of a PNE} \\
                            - General Game (Theorem \ref{THMgeneral})
                            \end{tabular}
                                                                    & \begin{tabular}{l}
                                                                    \textbf{Potential Function} \\
                                                                    - Singleton Game (Theorem \ref{THMpbsingletonIP})\\
                                                                    - Matroid Game (Theorem \ref{THMpbmatroidIP})
                                                                    \end{tabular}\\ \midrule
%%\textbf{Player-Specific}      	& \begin{tabular}{l}
\textbf{PS}      	& \begin{tabular}{l}
                            \textbf{Polynomial BR Dynamics} \\ 
                            - Singleton Game (Theorem \ref{THMpbsingletonCP})\\ 
                            - Matroid Game (Theorem \ref{THMpspbmatroidCP})
                            \end{tabular}
                                                                    &  \begin{tabular}{l}
                                                                    \textbf{Potential Function} \\
                                                                    - Two-Sided Singleton Market (Theorem \ref{THMtwosidedsingleton})\\
                                                                    - Two-Sided Matroid Market (Theorem \ref{THMtwosidedmatroid}) \\
%%                                                                       \textbf{Existence of a Pure Nash Equilibrium}   \\
                                                                       \textbf{Existence of a PNE}   \\
                                                                       - Singleton Game (Theorem \ref{THMpspbsingletonIP})\\
                                                                       - Matroid Game (Theorem \ref{THMpspbmatroidIP})
                                                                    \end{tabular}
\\ \bottomrule
\end{tabular}
\end{table}

The rest of the paper is organized as follows. 
We review previous results in Section \ref{SECpre}. 
Emphases are put on a formal description of the model and results of 
congestion games with priorities \cite{AGMRV08}. 
In Section \ref{SECourmodel}, 
we describe our model of priority-based congestion games. 
In Section \ref{SECsingleton}, 
we present some positive results on pure Nash equilibria in 
priority-based singleton congestion games. 
Section \ref{SECg-corr} is devoted to a description of 
how correlated two-sided markets with ties are generalized in our model. 
Finally, 
in Section \ref{SECbeyond}, 
we deal with priority-based congestion games which are not singleton games.

\section{Preliminaries}
\label{SECpre}

Let $\ZZ$ denote the set of the integers, 
and $\RR$ that of the real numbers. 
Subscripts $+$ and ${++}$ represent that the set consists of 
nonnegative numbers and positive numbers, 
respectively. 
\full{For instance, 
$\RR_+$ denotes the set of the nonnegative real numbers and 
$\ZZ_{++}$ that of the positive integers. }

\subsection{Congestion Games}
\label{SECconge}

%%A congestion game is formally described as follows. 
A congestion game is described by a tuple 
\full{$$
(N,E, (\mathcal{S}_i)_{i\in N}, (d_e)_{e\in E}). 
$$}%
\conf{$
(N,E, (\mathcal{S}_i)_{i\in N}, (d_e)_{e\in E})
$. }%
Here, 
$N = \{1,\ldots, n\}$ denotes the set of the players 
and $E$ that of the resources. 
Each player $i\in N$ has her \emph{strategy space} $\mathcal{S}_i\subseteq 2^E$, 
and chooses a \emph{strategy} $S_i\in \mathcal{S}_i$. 
The collection $(S_1,\ldots, S_n)$ of the chosen strategies is called a \emph{strategy profile}. 
For a resource $e\in E$ and a strategy profile $S = (S_1,\ldots, S_n)$, 
let $N_e(S)\subseteq N$ denote the set of players whose strategy includes $e$, 
and let $n_e(S) \in \ZZ_+$ denote the size of $N_e(S)$, 
i.e.,\  %
\conf{$N_e(S) = \{ i\in N \colon e \in S_i \}$ and $n_e(S)=|N_e(S)|$.}
\full{$$N_e(S) = \{ i\in N \colon e \in S_i \}, \quad n_e(S)=|N_e(S)|.$$ }

Each resource $e\in E$ has its \emph{delay function} $d_e\colon \ZZ_{++}\to \RR_+$. 
In a strategy profile $S$, 
the function value $d_e(n_e(S))$ represents the delay of a resource $e\in E$. 
\full{The objective of each player is to minimize 
her cost, 
which is the sum of the delays of the resources in her strategy. 
Namely, 
the cost $\gamma_i(S)$ imposed on a player $i\in N$ 
in a strategy profile $S$ 
is defined as 
$\gamma_i(S)=\sum_{e\in S_i}d_e(n_e(S))$, 
which is to be minimized. 
}%
\conf{The objective of each player $i\in N$ is to minimize 
her cost 
$\gamma_i(S)=\sum_{e\in S_i}d_e(n_e(S))$.}

\full{For a strategy profile $S=(S_1,\ldots, S_n)$ and a player $i\in N$, }
\conf{Let $S=(S_1,\ldots, S_n)$ be a strategy profile and $i\in N$.}
Let $S_{-i}$ denote a collection of the strategies in $S$ other than $S_i$, 
namely 
$S_{-i}=(S_1,\ldots,S_{i-1},S_{i+1},\ldots,  S_n)$. 
\full{A \emph{better response} of a player in a strategy profile is 
a change of her strategy so that her cost strictly decreases. 
Namely, 
when $i \in N$ changes her strategy from $S_i$ to $S_i'$ 
in a strategy profile $S$, 
it is a better response if 
% \begin{align*}
%     \gamma_i(S_{-i},S_i') < \gamma_i(S). 
% \end{align*}
$\gamma_i(S_{-i},S_i') < \gamma_i(S)$. }%
\conf{A \emph{better response} of $i$ in $S$ is 
to change her strategy from $S_i$ to $S_i'$ 
such that 
% \begin{align*}
%     \gamma_i(S_{-i},S_i') < \gamma_i(S). 
% \end{align*}
$\gamma_i(S_{-i},S_i') < \gamma_i(S)$.}
In particular, 
a better response from $S_i$ to $S_i'$ is a \emph{best response} if 
$S_i'$ minimizes $\gamma_i(S_{-i},S_i')$. 
A \emph{pure Nash equilibrium} is a strategy profile 
in which no player has a better response. 
\full{
Namely, 
a strategy profile $S$ is a pure Nash equilibrium 
if 
\begin{align*}
    \gamma_i(S) \le \gamma_i(S_{-i},S_i') \quad \mbox{for each player $i\in N$ and each of her strategy $S_i' \in \mathcal{S}_{i}$}. 
\end{align*}
}

A \emph{potential function} $\Phi$ is one 
which is 
defined on the set of the strategy profiles 
and satisfies %
\full{$$
\Phi(S_{-i}, S_i') - \Phi(S) = \gamma_i(S_{-i}, S_i') - \gamma_i(S)
$$}%
\conf{$
\Phi(S_{-i}, S_i') - \Phi(S) = \gamma_i(S_{-i}, S_i') - \gamma_i(S)
$}
for each strategy profile $S$, 
each player $i\in N$, 
and each strategy $S_i' \in \mathcal{S}$. 
The existence of a potential function implies the existence of a pure Nash equilibrium, 
because a strategy profile minimizing the potential function must be a pure Nash equilibrium. 
A game admitting a potential function is referred to as a \emph{potential game}.  
The following theorem is a primary result on congestion games, stating that 
each congestion game is a potential game and vice versa. 
\begin{theorem}[\cite{MS96,Ros73a}]
\label{THMconge}
    A congestion game is a potential game, 
    and hence possesses a pure Nash equilibrium. 
    Moreover, 
    every potential game is represented as a congestion game. 
\end{theorem}

Hereafter, we assume that each delay function $d_e$ ($e\in E$) is monotonically nondecreasing, 
i.e.,\ 
$d_e(x) \le d_e(x')$ if $x<x'$.

Study on congestion games from the viewpoint of \emph{algorithmic game theory} \cite{BV23book,NRTV07,Rou16} 
has appeared since around 2000. 
%%One fundamental line of research is to analyze the computability of pure Nash equilibria. 
% A \emph{singleton congestion game} is a class of congestion game 
% in which every strategy of every player is a singleton. 
For singleton congestion games, 
Ieong, McGrew, Nudelman, Shoham, and Sun \cite{IMNSS05} proved %
\full{that a pure Nash equilibrium in a singleton congestion game can be attained 
after a polynomial number of better responses. }
\conf{the following theorem.}

\begin{theorem}[\cite{IMNSS05}]
\label{THMsingleton}
In a singleton congestion game, 
starting from an arbitrary strategy profile, 
a pure Nash equilibrium is attained after a polynomial number of better %
\full{(hence, best)}%
responses. 
\end{theorem}

This theorem is followed by a large number of extensions.  
Recall that a \emph{player-specific congestion game} is one  
in which each resource $e \in E$ has a delay function $d_{i,e}\colon \ZZ_{++}\to \RR_+$ 
specific to each player $i\in N$. 
\full{Milchtaich \cite{Mil96} proved the following theorem for player-specific singleton congestion games. }

\begin{theorem}[\cite{Mil96}]
\label{THMpssingleton}
In a player-specific singleton congestion game, 
there exists a sequences of polynomial number of best responses 
starting from an arbitrary strategy profile and 
reaching a pure Nash equilibrium. 
\end{theorem}
Note that Theorem \ref{THMpssingleton} differs from Theorem \ref{THMsingleton} 
in that not any sequence of best responses reaches to a pure Nash equilibrium.

A significant work along this line is due to Ackermann, R\"{o}glin, and V\"{o}king \cite{ARV08,ARV09}, 
who employed the discrete structure of \emph{matroids} into congestion games. 
For a finite set $E$ and its subset family $\mathcal{S}\subseteq 2^E$, 
the pair $(E,\mathcal{S})$ is a \emph{matroid} if $\mathcal{S}\neq \emptyset$ and 
\begin{align}
    \label{EQaxiom}
    \mbox{for $S,S'\in \mathcal{S}$ and $e\in S \setminus S'$, 
    there exists $e' \in S'\setminus S$ such that
    $(S \setminus \{e\})\cup \{e'\} \in \mathcal{S}$.} 
\end{align}
A set in $\mathcal{S}$ is referred to as a \emph{base}. 
It follows from \eqref{EQaxiom} that all bases in $\mathcal{S}$ has the same cardinality, 
which 
is referred to as the \emph{rank} of the matroid $(E,\mathcal{S})$. 
% and denoted by $\rho(\mathcal{S})$. 

A congestion game $(N, E, (\mathcal{S}_i)_{i\in N}, (d_e)_{e\in E} )$ is referred to as a \emph{matroid congestion game}
if 
$(E,\mathcal{S}_i)$ is a matroid for every player $i\in N$. 
It is straightforward to see that a singleton congestion game is a spacial case of a matroid congestion game. 
Ackermann, R\"oglin, and V\"oking \cite{ARV08,ARV09} proved %
\conf{the following extensions of Theorems \ref{THMsingleton} and \ref{THMpssingleton}.}%
\full{the following extensions of Theorems \ref{THMsingleton} and \ref{THMpssingleton} to matroid congestion games.}

\begin{theorem}[\cite{ARV08}]
\label{THMmatroid}
In a matroid congestion game, 
starting from an arbitrary strategy profile, 
a pure Nash equilibrium is attained after a polynomial number of best responses. 
\end{theorem}

\begin{theorem}[\cite{ARV09}]
\label{THMpsmatroid}
In a player-specific matroid congestion game, 
there exists a sequence of polynomial number of better responses 
starting from an arbitrary strategy profile and 
reaching a pure Nash equilibrium. 
\end{theorem}

\full{
Since these works, 
matroid congestion games have been recognized as 
a well-behaved class of congestion games, 
and 
study on more generalized and related models followed. 
In the models of \emph{congestion games with mixed objectives} \cite{FLS18} and 
\emph{congestion games with complementarities} \cite{FLS17,Tak24}, 
the cost on a player is not necessarily the sum of the delays in her strategy. 
A \emph{budget game} \cite{DFRS19} is a variant of a congestion game, 
and their common generalization is proposed in \cite{KT22}. 
A \emph{resource buying game} \cite{HP15,Tak19} is another kind of a noncooperative game in which the players share the resources. 
In all of the above models, 
the fact that $(E,\mathcal{S}_i)$ is a matroid for each player $i$ plays a key role 
to guaranteeing the existence of a pure Nash equilibrium. }%
\conf{Since these works, 
matroid congestion games have been recognized as 
a well-behaved class of congestion games, 
and 
study on more generalized and related models followed \cite{DFRS19,FLS17,FLS18,HP15,KT22,Tak19,Tak24}.
In all of theses models, 
the fact that $(E,\mathcal{S}_i)$ is a matroid 
%%for each player $i$ 
plays a key role 
to guaranteeing the existence of a pure Nash equilibrium. }%
A further generalized model in which the strategy space is represented by a \emph{polymatroid} 
is studied in \cite{HKP18,HT18}. 
A different kind of relation between matroids and congestion games is investigated in  \cite{FGHPZ15}.

\subsection{Congestion Games with Priorities}
\label{SECAGMRV}

\full{Ackermann et al.\ \cite{AGMRV08} offered a model 
which commonly generalizes congestion games and a certain class of two-sided markets with ties. 
This model is described by a tuple} %
\conf{Ackermann et al.\ \cite{AGMRV08} offered a model 
which commonly generalizes congestion games and a certain class of two-sided markets with ties, 
which is 
referred to as a \emph{congestion game with priorities} and 
is described by a tuple $
(N,E, (\mathcal{S}_i)_{i\in N}, (p_e)_{e\in E}, (d_e)_{e\in E})
$. }%
\full{$$
(N,E, (\mathcal{S}_i)_{i\in N}, (p_e)_{e\in E}, (d_e)_{e\in E}), 
$$
in which 
the player set $N$, 
the resource set $E$, 
the strategy spaces $\mathcal{S}_i\subseteq 2^E$ ($i\in N$), 
and 
the delay functions $d_e\colon \ZZ_{++} \to \RR_+$ ($e\in E$) are 
the same as those in the classical model in Section \ref{SECconge}. }%
What is specific to this model is that 
each resource $e\in E$ has 
a \emph{priority function} $p_e\colon N \to \ZZ_{++}$. 
%%which represents the priorities of a resource $e$ over the players. 
If $p_e(i) < p_e(j)$ for players $i,j\in N$, 
then 
the resource $e$ prefers 
$i$ to $j$. 
%%which is reflected to the delay functions in the following way. 
%%

In a strategy profile $S=(S_1,\ldots, S_n)$,  
the delay of $e$ imposed on each player in $N_e(S)$ is 
determined in the following way. 
Define $p^*_e(S) \in \ZZ_{++} \cup \{+\infty\}$ by 
\begin{align*}
    p^*_e(S)=
    \begin{cases}
        \min \{ p_e(i) \colon i\in N_e(S)\} & \mbox{if $N_e(S)\neq \emptyset$}, \\
        +\infty                             & \mbox{if $N_e(S)= \emptyset$}. 
    \end{cases}
\end{align*}
For a positive integer $q$, 
define $n_e^q(S)\in \ZZ_+$ by
\begin{align}
    \label{EQeq}
    n_e^q(S) = \left|\{i \in N_e(S) \colon p_e(i)=q \} \right|.
\end{align}
Now the delay imposed on a player $i\in N_e(S)$ by the resource $e$ is defined as 
\begin{align*}
\begin{cases}
    d_e\left(n_e^{p_e^*(S)}(S)\right)   &\mbox{if $p_e(i) = p_e^*(S)$}, \\
    +\infty                             &\mbox{if $p_e(i) > p_e^*(S)$}.
\end{cases}
\end{align*}

\full{This model is referred to as a \emph{congestion game with priorities}. }
A special case in which all resources have the same priority function 
% i.e.,\ 
% $p_e = p_{e'}$ for each $e,e'\in E$, 
is called a \emph{congestion game with consistent priorities}. 
The general model is often referred to as a \emph{congestion game with inconsistent priorities}.

It is straightforward to see that 
\full{the model of 
congestion games with priorities includes congestion games. 
An instance $(N,E,(\mathcal{S}_i)_{i\in N},(d_e)_{e\in E})$ of a congestion game reduces to }%
a congestion game $(N,E,(\mathcal{S}_i)_{i\in N},(d_e)_{e\in E})$  reduces to 
a congestion game $(N,E,(\mathcal{S}_i)_{i\in N},(p_e)_{e\in E},(d_e)_{e\in E})$ with priorities 
in which 
all resources have the same constant priority 
function. 
As mentioned above, 
the model of congestion games with priorities also includes \emph{correlated two-sided markets with ties}. 
See Section \ref{SECcorrelated} for details.

\subsubsection{Singleton Games}

For singleton congestion games with consistent priorities, 
Ackermann et al.\ \cite{AGMRV08} proved the following theorem 
on the basis of Theorem \ref{THMsingleton}.

\begin{theorem}[\cite{AGMRV08}]
\label{THMsingletonCP}
In a singleton congestion game with consistent priorities, 
there exists a sequence of a polynomial number of best responses 
starting from an arbitrary strategy profile and 
reaching a pure Nash equilibrium. 
\end{theorem}

To the best of our knowledge, 
an extension of Theorem \ref{THMpssingleton} to
player-specific delay functions  
is missing in the literature, 
and will be discussed in a more generalized form in Section \ref{SECCP}. 
\full{

}%
Ackermann et al.\ \cite{AGMRV08} further proved that every singleton congestion game with inconsistent priorities 
is a potential game.

\begin{theorem}[\cite{AGMRV08}]
\label{THMsingletonIP}
A singleton congestion game with inconsistent priorities is a potential game, 
and hence possesses a pure Nash equilibrium. 
\end{theorem}

We remark that the potential function establishing Theorem \ref{THMsingletonIP} 
obeys a generalized definition of potential functions. 
\full{It maps a strategy profile to a sequence of vectors, which lexicographically decreases by a better response. }%
The details will appear in our proof of Theorem \ref{THMpbsingletonIP}, 
which extends Theorem \ref{THMsingletonIP} to priority-based singleton congestion games. 

For player-specific congestion games with inconsistent priorities, 
Ackermann et al.\ \cite{AGMRV08} designed a polynomial-time algorithm 
for constructing a pure Nash equilibrium. 
Let $n$ denote the number of the players 
and $m$ that of the resources. 

\begin{theorem}[\cite{AGMRV08}]
\label{THMpssingletonIP}
A player-specific singleton congestion game with inconsistent priorities possesses a pure Nash equilibrium, 
which can be computed in polynomial time with $O(n^3m)$ strategy changes. 
\end{theorem}

\subsubsection{Correlated Two-Sided Markets with Ties}
\label{SECcorrelated}

% A remarkable feature of the model of congestion games with priorities is 
% that it includes \emph{correlated two-sided markets with ties}.

\full{Here we describe a \emph{correlated to two-sided market with ties} \cite{AGMRV08},
and see that it can be represented as a player-specific congestion game with inconsistent priorities. }
For unity, 
we apply the terminology of congestion games to two-sided markets. 
For example, 
we use the terms players and resources instead of agents and markets. 
We also assume that the objective of a player is to minimize her delay, 
instead of to maximize her payoff. 

\full{A \emph{correlated two-sided market with ties} is represented by a tuple 
\begin{align*}
    (N,E,(\mathcal{S}_i)_{i\in N}, (c_{i,e})_{i\in N, e\in E}, (d_e)_{e\in E}). 
\end{align*}}%
\conf{A \emph{correlated two-sided market with ties} is represented as 
    $(N,E,(\mathcal{S}_i)_{i\in N}, (c_{i,e})_{i\in N, e\in E}, (d_e)_{e\in E})$. }%
\full{For each pair $(i,e)$ of a player $i\in N$ and a resource $e\in E$, }%
\conf{For each pair $(i,e) \in N\times E$, }%
a \emph{cost} $c_{i,e}\in \RR_+$ is associated. 
\full{The costs implicitly determine 
the preferences of the players, 
since the objective of a player is to minimize her cost. 
Moreover, 
each resource $e$ also prefer players with smaller costs, 
and in particular 
only accepts the players with smallest cost, 
which is formally described in the following way. }%
\conf{Both of the preferences of the players and the resources are determined by the costs in the following manner. }%

Let $S=(S_1,\ldots, S_n)$ be a strategy profile, 
and let $e\in E$ be a resource. 
Let $c^*_e(S)\in \RR_+$ be the minimum cost associated with a player in $N_e(S)$ and $e$, 
i.e.,\ 
$c_e^*(S)=\min\{ c_{i,e} \colon i\in N_e(S) \}$. 
Let $N_e^*(S)\subseteq N_e(S)$ denote the set of the players in $N_e(S)$ with cost $c_e^*(S)$, 
and 
let $|N_e^*(S)| = n_e^*(S)$. 
\full{Namely, 
\begin{align*}
    c^*_e(S) = \min\{c_{i,e} \colon i\in N_e(S)\}, 
    \quad 
    N_e^*(S) = \{i\in N_e(S) \colon c_{i,e}=c^*_e(S)\}, 
    \quad
    n_e^*(S) = |N_e^*(S)|.
\end{align*}}%
Each player in $N_e(S) \setminus N_e^*(S)$ receives an infinite cost from $e$. 
The cost on a player $i \in N_e^*(S)$ 
%%is determined by the delay function $d_e$ so that 
satisfies that 
\full{it is nonincreasing with respect to $n_e^*(S)$ 
and is equal to $c^*_e(S)$ if $n_e^*(S)=1$, 
i.e.,\ 
$i$ is the only player in $N_e^*(S)$. }%
\conf{it is equal to $c^*_e(S)$ if $n_e^*(S)=1$ 
and 
is nonincreasing with respect to $n_e^*(S)$. }%
This is represented by a bivariate delay function  $d_e:\RR_+\times \ZZ_{++}\to \RR_+$ 
such that, 
for each $x\in \RR_+$, 
$d_e(x, 1) = x$ 
and 
$d_e(x,y)$ is nondecreasing  with respect to $y$. 
In summary, 
the cost imposed on a player $i\in N_e(S)$ by $e$ is equal to 
\begin{align}
    \begin{cases}
        d_e(c_e^*(S), n_e^*(S))     &(i\in N_e^*(S)),\\
        +\infty                           &(i\in N_e(S)\setminus N_e^*(S)). 
    \end{cases}
\end{align}

A correlated two-sided market $(N,E, (\mathcal{S}_i), (c_{i,e}), (d_e))$ with ties 
reduces to a player-specific congestion game with inconsistent priorities. 
For a resource $e\in E$, 
construct a priority function $p_e\colon N\to \ZZ_{++}$ satisfying that 
\begin{align}
\label{EQpriority}
    \mbox{$p_e(i) < p_e(j)$ if and only if $c_{i,e} < c_{j,e}$ for each $i,j\in N$}.
\end{align}
Then, for each pair $(i,e)$ of a player $i\in N$ and a resource $e\in E$, 
define a player-specific delay function $d'_{i,e}\colon \ZZ_{++}\to \RR_+$ by 
\full{\begin{align*}
    d'_{i,e}(y) = d_e(c_{i,e}, y) \quad (y\in \ZZ_{++}). 
\end{align*}}%
\conf{$d'_{i,e}(y) = d_e(c_{i,e}, y)$ for each $y\in \ZZ_{++}$. }%

We refer to a correlated two-sided markets with ties 
in which each strategy of each player is a singleton 
as a 
\emph{correlated two-sided singleton markets with ties}. 
It follows from the above reduction that Theorem \ref{THMpssingletonIP} applies to a 
correlated two-sided singleton markets with ties. 
\full{Ackermann et al.\ \cite{AGMRV08} proved a stronger result that 
a correlated two-sided singleton markets with ties has a potential function. }%
\conf{Ackermann et al.\ \cite{AGMRV08} proved the following stronger result. }%

\begin{theorem}[\cite{AGMRV08}]
\label{THMcorrelated}
A correlated two-sided singleton market with ties is a potential game, 
and hence possesses a pure Nash equilibrium. 
\end{theorem}

\subsubsection{Extension to Matroid Games}

Finally, 
Ackermann et al.\ \cite{AGMRV08} provided the following extensions of 
Theorems \ref{THMsingletonCP}--\ref{THMcorrelated}
from singleton games to matroid games. 
For a matroid game, 
define its \emph{rank} $r$ as the 
the maximum rank 
of the matroids forming the strategy spaces of all players.

A 
better response of a player $i\in N$ 
in a strategy profile $S$ 
from a strategy $S_i$ to another strategy $S_i'$ is referred to as 
a \emph{lazy better response} if 
if there exists a 
sequence $(S^0_i , S^1_i,  . . . , S^k_i)$ of strategies of $i$ such that 
$S^0_i = S_i$, 
$S^{k}_i = S_i'$, 
$|S^{k'+1}_i \setminus S^{k'}_i | = 1$ and the cost on $i$ in 
a strategy profile $(S_{-i},S^{k'+1}_i)$ 
is strictly smaller than that in 
$(S_{-i},S^{k'}_i)$ 
for
each $k'=0,1,\ldots, k-1$. 
A \emph{potential game with respect to lazy better responses} 
is a game admitting a potential function 
which strictly decreases by a lazy better response. 

\begin{theorem}[\cite{AGMRV08}]
\label{THMmatroidCP}
    In a matroid congestion game with consistent priorities, 
there exists a sequence of a polynomial number of best responses 
starting from an arbitrary strategy profile and 
reaching a pure Nash equilibrium. 
\end{theorem}

\begin{theorem}[\cite{AGMRV08}]
\label{THMmatroidIP}
A matroid congestion game with inconsistent priorities is a potential game with respect to lazy better responses, 
and hence possesses a pure Nash equilibrium. 
\end{theorem}

\begin{theorem}[\cite{AGMRV08}]
\label{THMpsmatroidIP}
A player-specific matroid congestion game with inconsistent priorities possesses a pure Nash equilibrium, 
which can be computed in polynomial time with $O(n^3mr)$ strategy changes. 
\end{theorem}

\begin{theorem}[\cite{AGMRV08}]
\label{THMcorrelatedmatroid}
A correlated two-sided matroid market with ties is a potential game with respect to lazy better responses, 
and hence possesses a pure Nash equilibrium. 
\end{theorem}

\subsection{Priority-Based Affine Congestion Games}
\label{SECbv}

\full{In this subsection, 
we describe the model of priority-based affine congestion games with consistent priorities \cite{BV23tcs}, by using the terminology of 
congestion games with priorities. }%
\full{A priority-based affine congestion game with consistent priorities is described by a tuple 
$$
(N, E, (\mathcal{S}_i)_{i\in N}, p, (\alpha_e, \beta_e)_{e\in E}). 
$$}%
\conf{A priority-based affine congestion game with consistent priorities \cite{BV23tcs} is described as
$
(N, E, (\mathcal{S}_i)_{i\in N}, p, (\alpha_e, \beta_e)_{e\in E})$. }%
\full{Again, 
$N$ and $E$ denote the set of the players and that of the resources, respectively, 
and 
$\mathcal{S}_i \subseteq 2^E$ is the strategy space of a player $i\in N$. }%
Note that all resources 
have the same priority function $p \colon N \to \ZZ_{++}$. 
Each resource $e\in E$ is associated with two nonnegative real numbers $\alpha_e,\beta_e\in \RR_+$, 
which determine the delay function of $e$ in the following manner. 

Let $S$ be a strategy profile, 
$e\in E$ be a resource, and 
$q \in \ZZ_+$ a positive integer.  
Define 
$n_e^q(S) \in \ZZ_+$ 
as in \eqref{EQeq}, 
in which $p_e$ is replaced by $p$. 
Similarly, 
define 
%%$n_e^{<q}(S),n_e^{\le q}(S) \in \ZZ_+$ by 
$n_e^{<q}(S)\in \ZZ_+$ by 
\begin{align}
    \label{EQl}
    n_e^{< q}(S)  \   {}&{}=\ \left|\{i \in N_e(S) \colon p(i) < q \} \right|.
    % \\
    % \label{EQleq}
    % n_e^{\le q}(S)\   {}&{}=\ \left|\{i \in N_e(S) \colon p(i) \le q \} \right|. 
\end{align}
Now the delay imposed on a player $i\in N_e(S)$ by $e$ is 
defined as 
\begin{align}
    \label{EQBV}
    \alpha_e \cdot \left(n_e^{<p(i)}(S) + 
    %%\frac{1}{2}\left(n_e^{p(i)}(S)+1\right) 
    \frac{n_e^{p(i)}(S)+1}{2}
    \right) + \beta_e, 
\end{align}
which is interpreted in the following way. 
The delay imposed on player $i\in N_e(S)$ by $e\in E$ is affected by 
the $n_e^{<p(i)}(S)$ players in $N_e(S)$ more prioritized than $i$. 
It is also affected by the $n_e^{p(i)}(S)$ players with the same priority as $i$, 
which is reflected to 
$(n_e^{p(i)}(S)+1)/2$ in \eqref{EQBV}. 
This value is the expected 
%%rank of $i$ among the $n_e^{p(i)}(S)$  players when they are ranked randomly. 
number of the players more or equally prioritized than $i$ when the ties of the $n_e^{p(i)}(S)$  players are broken 
uniformly at random.  

\full{
Bil\`o and Vinci \cite{BV23tcs} proved that every 
priority-based affine congestion game with consistent priorities has a pure Nash equilibrium, 
and that it 
can be constructed by 
finding a pure Nash equilibrium of the most prioritized players, 
and 
then 
inductively 
extending the pure Nash equilibrium of the players with up to the $k$-th priority 
to those with up to the $(k+1)$-st priority. 
In each step, 
the game restricted to the players with the $(k+1)$-st priority is a potential game. 
}

%%\begin{theorem}[Bil\`o and Vinci \cite{BV23tcs}]
\begin{theorem}[\cite{BV23tcs}]
    A priority-based affine congestion game with consistent priorities 
    possesses a pure Nash equilibrium. 
\end{theorem}

\full{
We should remark that 
Bil\`o and Vinci \cite{BV23tcs} further conducted an elaborated analysis on 
the price of anarchy and the price of stability of the pure Nash equilibria of this model, 
which might be a main contribution of their paper. 
}

% The objective of this paper is to a reasonable answer to the open question of Ackermann et al.\ \cite{AGMRV08}. 
% Namely, 
% we provide a model generalizing congestion games with priorities 
% in which each player $i$ is imposed a finite delay by the more prioritized players sharing a resource with $i$. 

\section{Our Model}
\label{SECourmodel}

We first point out that 
the model of Bil\`o and Vinci \cite{BV23tcs} described in Section \ref{SECbv} partially 
answers to the open question of Ackermann et al.\ \cite{AGMRV08}. 
Indeed, 
the delay \eqref{EQBV} of a player $i \in N_e(S)$ is finitely affected by 
the more prioritized players in $N_e(S)$. 
Meanwhile, 
compared to the model of Ackermann et al.\ \cite{AGMRV08}, 
the delay \eqref{EQBV} is specific in that it is a particular affine function 
of $n_e^{<p(i)}(S)$ and $n_e^{p(i)}(S)$, and 
the priorities of the resources are consistent. 
Below we resolve these points by providing a common generalization of the two models, 
which provides a full answer to the open question in \cite{AGMRV08}. 

Our model is represented 
\full{by a tuple $$
(N, E, (\mathcal{S}_i)_{i\in N}, (p_e)_{e\in E}, (d_e)_{e\in E}), 
$$}%
\conf{as $
(N, E, (\mathcal{S}_i)_{i\in N}, (p_e)_{e\in E}, (d_e)_{e\in E}), 
$ }%
which is often abbreviated as $(N, E, (\mathcal{S}_i), (p_e), (d_e))$. 
\full{Again, 
$N$ and $E$ denote the sets of players and resources, respectively, 
each player $i\in N$ has her strategy space $\mathcal{S}_i\subseteq 2^E$, 
and 
each resource $e\in E$ has a priority function $p_e \colon N \to \ZZ_{++}$. 

}%
Let $S = (S_1, \ldots, S_n)$ be a strategy profile, 
$i\in N$, and $e\in S_i$. 
Reflecting the delay function \eqref{EQBV}, 
\full{our delay function $d_e\colon \ZZ_+ \times \ZZ_{++} \to \RR_+$ ($e\in E$) is a bivariate function with }%
\conf{our delay function $d_e\colon \ZZ_+ \times \ZZ_{++} \to \RR_+$ is a bivariate function with }
variables $n_e^{<p_e(i)}(S)$ and $n_e^{p_e(i)}(S)$. 
Namely, 
the delay imposed on $i$ by $e$ is 
\full{described as }
\begin{align}
\label{EQdelay}
    d_{e}\left( n_{e}^{<p_{e}(i)}(S), n_{e}^{p_{e}(i)}(S) \right). 
\end{align}
% This means that 
% the players in $N_e(S)$ with better and the same priorities differently affect the delay on $i$. 
% It is straightforward to see that 
% this delay function \eqref{EQdelay} generalizes the delay function \eqref{EQBV} of a priority-based affine congestion game. 

We assume that each delay function $d_e$ ($e\in E$) has the following properties: 
\begin{alignat}{2}
\label{EQmonotonex}
&d_e(x,y) \le d_e(x',y)& \quad  &(\mbox{if $x<x'$}), \\
\label{EQmonotoney}
&d_e(x,y) \le d_e(x,y')& \quad  &(\mbox{if $y<y'$}), \\
\label{EQplusone}
&d_e(x,y) \le d_e(x+y-1,1)& \quad &(\mbox{for each $x\in \ZZ_+$ and $y\in \ZZ_{++}$}).
\end{alignat}
Property \eqref{EQmonotonex} and \eqref{EQmonotoney}  
mean that the delay function $d_e$ is nondecreasing with respect to 
$n_{e}^{<p_{e}(i)}(S)$ and 
$n_{e}^{p_{e}(i)}(S)$, 
respectively. 
\full{These properties reflect the monotonicity of the delay functions in the previous models. }
Property \eqref{EQplusone} means that 
the cost on $i$ increases 
if the $n_e^{p_e(i)}(S)-1$ players in $N_e(S)$ with the same priority as $i$ 
are replaced by the same number of more prioritized players. 
This property captures the characteristic of the models of \cite{AGMRV08,BV23tcs} 
that 
prioritized players produce more delays than those with the same priority. 

We refer to our model as a \emph{priority-based congestion game with inconsistent priorities}, 
or \emph{priority-based congestion game} for short. 
If the resources have the same priority function, 
then the game is referred to as a \emph{priority-based congestion game with consistent priorities}.

A priority-based affine congestion game 
$(N, E, (\mathcal{S}_i)_{i\in N}, p, (\alpha_e, \beta_e)_{e\in E})$ with consistent priority \cite{BV23tcs} 
is represented as a priority-based congestion game
with consistent priorities $p$ and 
delay function $d_e\colon \ZZ_+ \times \ZZ_{++} \to \RR_+$ ($e\in E$) defined as in 
\eqref{EQBV}, 
namely 
\begin{align}
\label{EQreduceBV}
    d_{e}\left( n_{e}^{<p_{e}(i)}(S), n_{e}^{p_{e}(i)}(S) \right) = 
    \alpha_e \left(n_e^{<p(i)}(S) + 
    %%\frac{1}{2}\left(n_e^{p(i)}(S)+1\right) 
    \frac{n_e^{p(i)}(S)+1}{2}
    \right) + \beta_e. 
\end{align}
It is not difficult to see that the delay function $d_e$ in \eqref{EQreduceBV} satisfies 
the properties \eqref{EQmonotonex}--\eqref{EQplusone}. 

A congestion game with inconsistent priorities \cite{AGMRV08} is also a special case of a priority-based congestion game. 
Given a congestion game 
$(N,E, (\mathcal{S}_i)_{i\in N}, (p_e)_{e\in E}, (d_e)_{e\in E})$ with priorities, 
define a delay function $d_e'\colon \ZZ_+\times \ZZ_{++}\to \RR$ of 
a priority-based congestion game by 
\begin{align}
\label{EQreduceA}
    d_e'(x,y) = 
    \begin{cases}
        +\infty  & (x \ge 1), \\
        d_e(y)  & (x =0).
    \end{cases}
\end{align}
Again, 
the delay function $d_e'$ in \eqref{EQreduceA} satisfies 
the properties \eqref{EQmonotonex}--\eqref{EQplusone} 
if $d_e$ is a nondecreasing function. 
The properties \eqref{EQmonotonex} and \eqref{EQmonotoney} are directly derived. 
The property \eqref{EQplusone} follows from the fact that 
$d_e'(x+y-1,1) \neq + \infty$ only if 
%%$x+y-1=0$, 
%%which holds only if 
$(x,y)=(0,1)$, 
and in that case 
both $d_e'(x,y)$ and $d_e'(x+y-1,1)$ is equal to $d_e'(0,1)= d_e(1)$.

\section{Priority-Based Singleton Congestion Games}
\label{SECsingleton}

\full{In this section, 
we present some theorems on pure Nash equilibria 
in singleton games in our model. 
}

\subsection{Consistent Priorities}
\label{SECCP}

\full{
In this subsection, 
we present 
a theorem on pure Nash equilibria in priority-based player-specific singleton congestion games 
with consistent priorities
(Theorem \ref{THMpbsingletonCP}). 
This theorem is not only an extension of Theorems \ref{THMsingleton} and \ref{THMsingletonCP}, 
which concern pure Nash equilibria in non-player-specific singleton congestion games, 
but also implies the existence of pure Nash equilibria 
in player-specific congestion games with consistent priorities (Corollary \ref{COR}), 
which is missing in the literature. 
}

\conf{
In a singleton game, 
a strategy $\{e\}$ is simply denoted by $e$. 
The following theorem is not only an extension of Theorems \ref{THMsingleton} and \ref{THMsingletonCP}, 
which concern pure Nash equilibria in non-player-specific singleton congestion games, 
but also implies the existence of pure Nash equilibria 
in player-specific congestion games with consistent priorities (Corollary \ref{COR}), 
which is missing in the literature. 
}

Hereafter, 
some theorems are marked with $(\star)$, meaning that their proofs appear in Appendix. 

\begin{theorem}[$\star$]
\label{THMpbsingletonCP}
In a priority-based player-specific singleton congestion game $G=(N, E, (\mathcal{S}_i), p, (d_{i,e}))$ 
with consistent priorities, 
there exists a sequence of polynomial number of better responses 
starting from an arbitrary strategy profile and 
reaching a pure Nash equilibrium. 
\end{theorem}

\full{
The following corollary is a direct consequence of Theorem \ref{THMpbsingletonCP}. 
}

\begin{corollary}
\label{COR}
In a player-specific singleton congestion game with consistent priorities, 
there exists a sequences of polynomial number of better responses 
starting from an arbitrary strategy profile and 
reaching a pure Nash equilibrium. 
\end{corollary}

\begin{remark}
    Corollary \ref{COR} does not imply that 
    a pure Nash equilibrium  
    \full{in a priority-based singleton congestion games which is not player-specific }%
    is obtained from an \emph{arbitrary} sequence of best responses. 
    This is because, 
    as described in the proof for Theorem \ref{THMpbsingletonCP}, 
    the order of the players in the sequence is specified by the priority function. 
\end{remark}

\subsection{Inconsistent Priorities}

\full{In this subsection, 
we investigate priority-based singleton congestion games with inconsistent priorities. }%
We first prove the following extension of Theorem \ref{THMsingletonIP}. 

\begin{theorem}
\label{THMpbsingletonIP}
A priority-based singleton congestion game 
$(N,E, (\mathcal{S}_i), (p_{e}), (d_e))$ 
with inconsistent priorities 
is a potential game, 
and hence possesses a pure Nash equilibrium. 
\end{theorem}

\begin{proof}
For each strategy profile $S=(e_1,\ldots, e_n)$, 
define its potential $\Phi(S) \in (\RR_+ \times \ZZ_{++})^n$ as follows. 
Let 
$e\in E$ be a resource, 
and let 
$Q_e(S)=\{ q_1,..., q_{k^*}\}$ be a set of integers such that 
$Q_e(S)=\{ q \colon n_e^ {q}(S) > 0\}$  and $ q_1 < \cdots <  q_{k^*}$. 
The resource $e\in E$ contributes the following 
$n_e(S)$ vectors in $\RR_+ \times \ZZ_{++}$
to $\Phi(S)$: 
\full{
\begin{align}
&{}(d_e(0,1),\,  q_1), \ldots, (d_e(0,n_e^{ q_1}(S)),\,  q_1), \notag\\
&{}(d_e(n_e^{ q_1}(S),1),\,  q_2), \ldots, (d_e(n_e^{ q_1}(S),n_e^{q_2}(S)),\,  q_2), \notag\\
%%&{}\vdots \notag\\
&{}\ldots, \notag\\
&{}\left(d_e\left(n_e^{< q_k}(S), 1 \right),\,  q_k\right), \ldots, 
\left( d_e\left(n_e^{< q_k}(S), n_e^{ q_k}(S)\right),\,  q_k\right), \notag\\
%%&{}\vdots \notag\\
&{}\ldots, \notag\\
&{}\left(d_e\left(n_e^{< q_{k^*}}(S), 1 \right),\,  q_{k^*}\right), \ldots, 
\left( d_e\left(n_e^{<{ q_{k^*}}}(S), n_e^{ q_{k^*}}(S)\right),\,  q_{k^*}\right). 
\label{EQpotentialvectors}
\end{align}
}%
\conf{
\begin{align}
&{}(d_e(0,1),\,  q_1), \ldots, (d_e(0,n_e^{ q_1}(S)),\,  q_1), 
(d_e(n_e^{ q_1}(S),1),\,  q_2), \ldots, (d_e(n_e^{ q_1}(S),n_e^{q_2}(S)),\,  q_2), \notag\\
%%&{}\vdots \notag\\
&{}\ldots, 
(d_e(n_e^{< q_k}(S), 1 ),\,  q_k), \ldots, 
( d_e(n_e^{< q_k}(S), n_e^{ q_k}(S)),\,  q_k), \notag\\
%%&{}\vdots \notag\\
&{}\ldots, 
(d_e(n_e^{< q_{k^*}}(S), 1 ),\,  q_{k^*}), \ldots, 
( d_e(n_e^{<{ q_{k^*}}}(S), n_e^{ q_{k^*}}(S)),\,  q_{k^*}). 
\label{EQpotentialvectors}
\end{align}
}

For two vectors $(x,y),(x',y')\in \RR_+ \times \ZZ_{++}$, 
we define a lexicographic order $(x,y) \preceqlex (x',y')$ if 
\begin{align*}
&{}x < x', \quad \mbox{or} \quad
x=x'  \mbox{ and }  y\leq y'. 
\end{align*}
The strict relation 
$(x,y) \preclex (x',y')$ means that 
$(x,y) \preceqlex (x',y')$ and $(x,y) \neq (x',y')$ hold. 

The potential $\Phi(S)$ is obtained by ordering the $n$ vectors 
contributed by all resources 
in the lexicographically nondecreasing order. 
We remark that 
the order in \eqref{EQpotentialvectors} 
%%of the vectors contributed by $e$ shown in \eqref{EQpotentialvectors} 
is lexicographically nondecreasing, 
which can be derived from 
%%the properties 
\eqref{EQmonotonex}--\eqref{EQplusone} as follows. 
It follows from \eqref{EQmonotoney} that
\begin{align}
d_e(
n_e^{< q_k}(S), y
) 
\le 
d_e(
n_e^{< q_k}(S), y+1
) 
\quad 
\mbox{($k=1,\ldots, k^*$, $y=1,\ldots,  q_{k-1}$)}, 
\end{align}
and 
from \eqref{EQmonotonex} and \eqref{EQplusone} that 
\begin{align}
\label{EQkplus1}
\full{
d_e\left(
n_e^{< q_k}(S), n_e^{ q_k}(S)
\right) 
\le 
d_e\left(
n_e^{< q_{k+1}}(S) -1,1
\right) 
\le 
d_e\left(
n_e^{< q_{k+1}}(S),1\right) \quad \mbox{($k=1,\ldots, k^*-1$)}. 
}
\conf{d_e(
n_e^{< q_k}(S), n_e^{ q_k}(S)
) 
\le 
d_e(
n_e^{< q_{k+1}}(S) -1,1
) 
\le 
d_e(
n_e^{< q_{k+1}}(S),1) \ \mbox{($k=1,\ldots, k^*-1$)}. 
}
\end{align}

We then define a lexicographic order over the potentials.  
For strategy profiles $S$ and $S'$, 
where 
\begin{align*}
    \Phi(S) = ((x_1,y_1),\ldots, (x_n,y_n)), \quad 
    \Phi(S') = ((x'_1,y'_1),\ldots, (x'_n,y'_n)), 
\end{align*}
define $\Phi(S') \preceqlex \Phi(S)$ if 
there exists an integer $\ell$ with $1\le \ell \le n$ such that 
\full{\begin{align*}
    &{}\mbox{$(x'_{\ell'},y'_{\ell'}) = (x_{\ell'},y_{\ell'})$ for each $\ell' < \ell$, and  
    $(x'_{\ell},y'_{\ell}) \preclex (x_{\ell},y_{\ell})$}.
\end{align*}
}%
\conf{$(x'_{\ell'},y'_{\ell'}) = (x_{\ell'},y_{\ell'})$ for each $\ell' < \ell$ and  
    $(x'_{\ell},y'_{\ell}) \preclex (x_{\ell},y_{\ell})$. }% 
The strict relation 
$\Phi(S') \preclex \Phi(S)$ means that 
$\Phi(S') \preceqlex \Phi(S)$ and $\Phi(S') \neq \Phi(S)$ hold. 

\full{
Suppose that a player $i$ has a better response 
in a strategy profile $S$, 
which changes her strategy from $e$ to $e'$. 
Let $S'=(S_{-i}, e')$. 
Below we show that 
$\Phi(S') \preclex \Phi(S)$, 
which completes the proof. 
}%
\conf{Suppose that $i\in N$ has a better response 
in a strategy profile $S$, 
changing her strategy from $e$ to $e'$. 
Let $S'=(S_{-i}, e')$. 
Below we show 
$\Phi(S') \preclex \Phi(S)$, 
completing the proof. 
}

Let $p_e(i) =  q$ 
and $p_{e'}(i) =  q'$. 
Since the delay imposed on $i$ becomes smaller 
due to the better response, 
it holds that 
\begin{align}
\label{EQbr}
d_{e'}(
    n_{e'}^{< q'}(S), n_{e'}^{ q'}(S)+1 
)
< 
d_e(
n_e^{< q}(S), n_e^{ q}(S)
). 
\end{align}
Note that $e'$ contributes a vector 
\begin{align}
\label{EQnewvector}
( 
d_{e'}(
    n_{e'}^{< q'}(S), n_{e'}^{ q'}(S)+1 
),\, 
 q' ) 
\end{align}
to $\Phi(S')$ but not to $\Phi(S)$. 
To prove $\Phi(S') \preclex \Phi(S)$, 
it suffices to 
show that 
a vector belonging to $\Phi(S)$ but not to $\Phi(S')$ is 
lexcographically larger than the vector \eqref{EQnewvector}.

First, 
consider the vectors in $\Phi(S)$ contributed by $e$. 
Let $Q_e(S)=\{q_1,\ldots, q_{k^*}\}$, 
where $q_1 < \cdots < q_{k^*}$. 
The better response of $i$ changes 
the vectors in $\Phi(S)$ whose second component is larger than $ q$, 
because the first argument of the delay function $d_e$ decreases by one. 
If $q=q_{k^*}$, 
then those vectors do not exist and thus we are done. 
Suppose that $q=q_k$ for some $k<k^*$. 
Among those vectors, 
the lexicographically smallest one is 
\full{
$$
\left( 
    d_e\left(
        n_e^{<  q_{k+1}}(S), 1
    \right),\, 
     q_{k+1}
\right). 
$$}%
\conf{$
( 
    d_e(
        n_e^{<  q_{k+1}}(S), 1
    ),\, 
     q_{k+1}
)
$.}
Recall \eqref{EQkplus1}, 
saying that 
\full{$$
d_e\left(
n_e^{< q_k}(S), n_e^{ q_k}(S)
\right) 
\le 
d_e\left(
n_e^{< q_{k+1}}(S),1
\right) ,
$$}%
\conf{$
d_e(
n_e^{< q_k}(S), n_e^{ q_k}(S)
) 
\le 
d_e(
n_e^{< q_{k+1}}(S),1
)$, }%
and thus 
\full{$$
d_{e'}\left(
    n_{e'}^{< q'}(S), n_{e'}^{ q'}(S)+1 
\right)
< 
d_e\left(
n_e^{< q_{k+1}}(S),1
\right) 
$$}%
\conf{$
d_{e'}(
    n_{e'}^{< q'}(S), n_{e'}^{ q'}(S)+1 
)
< 
d_e(
n_e^{< q_{k+1}}(S),1
) 
$ }%
follows from \eqref{EQbr}. 
Hence, 
we conclude that 
\full{$$
\left( 
d_{e'}\left(
    n_{e'}^{< q'}(S), n_{e'}^{ q'}(S)+1 
\right),\,  q' \right) 
\preclex
\left( 
    d_e\left(
            n_e^{< q_{k+1}}(S),1
    \right),\, 
     q_{k+1}
\right). 
$$
}%
\conf{$
( 
d_{e'}(
    n_{e'}^{< q'}(S), n_{e'}^{ q'}(S)+1 
),\,  q' ) 
\preclex
( 
    d_e(
            n_e^{< q_{k+1}}(S),1
    ),\, 
     q_{k+1}
)$. 
}%

Next, 
consider the vectors in $\Phi(S)$ contributed by $e'$. 
Without loss of generality, 
suppose that 
%%those vectors exist in $\Phi(S)$, 
%%i.e.,\ 
there exists a positive integer $q''$ 
such that 
%%$n_{e'}^{ q''}(S) > 0$ 
$q'' \in Q_{e'}(S)$
and $q'' > q'$. 
Let $q''$ be the smallest integer satisfying these conditions. 
% Due to the strategy change of $i$, 
% the vectors whose second component is larger than $ q'$ changes. 
The lexicographically smallest vector in $\Phi(S)$ contributed by $e'$ and changed by the better response of $i$ is 
\full{$$
\left(
    d_{e'}\left(
        n_{e'}^{<  q''}(S) , 1 
    \right),\,
     q''
\right). 
$$}%
\conf{$
(
    d_{e'}(
        n_{e'}^{<  q''}(S) , 1 
    ),\,
     q''
)$. }%
It follows from the property \eqref{EQplusone} that 
\full{$$
d_{e'}\left( 
    n_{e'}^{< q'}(S), n_{e'}^{ q'}(S)+1
\right)
\le
    d_{e'}\left(
        n_{e'}^{<  q''}(S) , 1 
    \right),
$$}%
\conf{$
d_{e'}( 
    n_{e'}^{< q'}(S), n_{e'}^{ q'}(S)+1
)
\le
    d_{e'}(
        n_{e'}^{<  q''}(S) , 1 
    )$, }%
and thus 
\full{$$
\left( 
d_{e'}\left(
    n_{e'}^{< q'}(S), n_{e'}^{ q'}(S)+1 
\right),\,  q' \right) 
\preclex 
\left(
    d_{e'}\left(
        n_{e'}^{<  q''}(S) , 1 
    \right),\,
     q''
\right), 
$$}%
\conf{$
( 
d_{e'}(
    n_{e'}^{< q'}(S), n_{e'}^{ q'}(S)+1 
),\,  q' ) 
\preclex 
(
    d_{e'}(
        n_{e'}^{<  q''}(S) , 1 
    ),\,
     q''
)$, }%
completing the proof. 
\end{proof}

\full{
We next show the following theorem, 
which corresponds to Theorem \ref{THMpssingletonIP} 
but 
does not include a polynomial bound on the number of strategy changes. 
}
\conf{The next theorem  corresponds to Theorem \ref{THMpssingletonIP}, 
while it does not include a polynomial bound on the number of strategy changes. 
}

\begin{theorem}[$\star$]
\label{THMpspbsingletonIP}
A priority-based player-specific singleton congestion game 
with inconsistent priorities 
possesses a pure Nash equilibrium, 
which can be computed with a finite number of strategy changes. 
\end{theorem}

%%\subsection{Player-Specific Games}

\section{Generalized Correlated Two-Sided Markets with Ties}
\label{SECg-corr}

%%As mentioned in Section \ref{SECcorrelated}, 
%%the model of correlated two-sided markets with ties 
%%is a special class of player-specific congestion games with inconsistent priorities 
%%\cite{AGMRV08}. 
In this section, 
we introduce the model of \emph{generalized correlated two-sided markets with ties}, 
which 
generalizes correlated two-sided markets with ties described in Section \ref{SECcorrelated}. 
We show that this model is a special class of priority-based player-specific congestion games with inconsistent priorities, 
and 
it includes priority-based congestion games with inconsistent priorities. 
This is in contrast to the situation of correlated two-sided markets with ties, 
in which it is unclear whether correlated two-sided markets with ties include 
congestion games with inconsistent priorities.  
We then prove that a generalized correlated two-sided market with ties is a potential game, 
which extends Theorem \ref{THMcorrelated}. 

\full{
\subsection{Model}
}

\full{
A generalized correlated two-sided market with ties is described by a tuple 
$$
(N,E, (\mathcal{S}_i)_{i\in N},(c_{i,e})_{i\in N, e\in E}, (d_e)_{e\in E}  ).
$$
Again, 
$N$ and $E$ denote the sets of the players and resources, respectively. 
For each player $i\in N$ and each resource $e\in E$, 
a nonnegative real number $c_{i,e}\in \RR_+$ is associated, 
which implies the preferences of $i$ and $e$, 
and 
are reflected in the delay function $d_e$ of $e$ in the following way. 
}

\conf{A generalized correlated two-sided market with ties is described as $
(N,E, (\mathcal{S}_i)_{i\in N},(c_{i,e})_{i\in N, e\in E}, (d_e)_{e\in E}  )$.
For each pair $(i,e) \in N\times E$, 
a cost $c_{i,e}\in \RR_+$ is associated, 
which implies the preferences of $i$ and $e$, 
and 
are reflected in the delay function $d_e$ of $e$ in the following way. 
}

Let $S =(S_1,\ldots, S_n)$ be a strategy profile 
and $e\in E$ be a resource. 
In the same way as \eqref{EQeq} and \eqref{EQl}, 
for a nonnegative number $ q \in \RR_+$, 
define 
\full{
$n_e^ q(S), n_e^{< q}(S) \in \ZZ_+$ by 
\begin{align*}
    n_e^ q(S)    \   {}&{}=\ \left|\{i \in N_e(S) \colon c_{i,e} =  q\}\right|, \\
    n_e^{< q}(S)  \   {}&{}=\ \left|\{i \in N_e(S) \colon c_{i,e} <  q \} \right|.
\end{align*}
}%
\conf{$n_e^ q(S) =|\{i \in N_e(S) \colon c_{i,e} =  q\}|$ and  
    $n_e^{< q}(S) = |\{i \in N_e(S) \colon c_{i,e} <  q \} |$. }%
Note that $n_e^{c_{i,e}}(S) >0$ 
if $e\in S_i$. 
The delay function $d_e$ is a trivariate function $d_e\colon \RR_+ \times \ZZ_+ \times \ZZ_{++} \to \RR_+$. 
The cost imposed by $e$ on a player $i\in N_e(S)$ is 
\full{$$
d_e\left(c_{i,e},\, n_e^{<c_{i,e}}(S),\, n_e^{c_{i,e}}(S)\right). 
$$}%
\conf{$
d_e(c_{i,e},\, n_e^{<c_{i,e}}(S),\, n_e^{c_{i,e}}(S))$. }%

Here, 
the delay functions $d_e$ ($e \in E$) have the following properties: 
\begin{align}
\label{EQmonotonec3}
&d_e(c,x,y) \le d_e(c',x,y)& \quad  &(\mbox{if $c<c'$}), \\
\label{EQmonotonex3}
&d_e(c,x,y) \le d_e(c,x',y)& \quad  &(\mbox{if $x<x'$}), \\
\label{EQmonotoney3}
&d_e(c,x,y) \le d_e(c,x,y')& \quad  &(\mbox{if $y<y'$}), \\
\label{EQplusone3}
&d_e(c,x,y) \le d(c,x+y-1,1)& \quad &(\mbox{for each $x\in \ZZ_+$ and $y\in \ZZ_{++}$}).
\end{align}
The properties \eqref{EQmonotonec3}--\eqref{EQmonotoney3} represent the monotonicity of $d_e$, 
while \eqref{EQmonotonex3}--\eqref{EQplusone3} corresponds to 
the properties \eqref{EQmonotonex}--\eqref{EQplusone} of the delay functions in priority-based congestion games. 
We also remark that $d_e(c,0,1)$ is not necessarily equal to $c$, 
whereas $d_e(c,1)=c$ in correlated two-sided markets with ties. 
% The total cost imposed on the player $i$ in the strategy profile $S$ is 
% $$
% \sum_{e\in S_i}
% d_e\left(c_{i,e},\, n_e^{<c_{i,e}}(S),\, n_e^{c_{i,e}}(S)\right). 
% $$

\full{
\subsection{Relation to Other Models}
}
\full{
A correlated two-sided market with ties $(N,E,(\mathcal{S}_i), (c_{i,e}), (d_e))$ is represented as a 
generalized correlated two-sided market with ties $(N,E,(\mathcal{S}_i), (c_{i,e}), (d_e'))$ 
by defining the trivariate function $d_e'\colon \RR_+\times\ZZ_{+}\times \ZZ_{++}\to \RR_+$ by 
%%For each delay function $d_e\colon \RR_+\times \ZZ_{++}\to \RR_+$ , 
%%construct a 
\begin{align*}
    d_e'(c,x,y) = 
    \begin{cases}
        d_e(c,y)    & \mbox{if $x=0$}, \\
        +\infty     & \mbox{if $x\ge 1$} 
    \end{cases}
\end{align*}
for each resource $e\in E$. 
}

\conf{A correlated two-sided market $(N,E,(\mathcal{S}_i), (c_{i,e}), (d_e))$ with ties reduces to a 
generalized correlated two-sided market $(N,E,(\mathcal{S}_i), (c_{i,e}), (d_e'))$ with ties 
in which $d_e'$ is defined by 
%%For each delay function $d_e\colon \RR_+\times \ZZ_{++}\to \RR_+$ , 
%%construct a 
\begin{align*}
    d_e'(c,x,y) = 
    \begin{cases}
        d_e(c,y)    & \mbox{if $x=0$}, \\
        +\infty     & \mbox{if $x\ge 1$} 
    \end{cases} 
    \quad \mbox{for each resource $e\in E$.}
\end{align*}
}

The following propositions show that 
generalized correlated two-sided markets with ties 
lie between 
priority-based congestion games with inconsistent priorities 
and 
priority-based player-specific congestion games with inconsistent priorities.

\begin{proposition}[$\star$]
\label{PROPnonps}
A priority-based congestion games with inconsistent priorities 
is represented as a generalized correlated two-sided market with ties. 
\end{proposition}

\begin{proposition}[$\star$]
\label{PROPps}
A generalized correlated two-sided market with ties is represented as a 
priority-based player-specific congestion games with inconsistent priorities. 
\end{proposition}

\full{
\subsection{Pure Nash Equilibria and Potential}
}

From Proposition \ref{PROPps} and Theorem \ref{THMpspbsingletonIP}, 
it follows that 
a generalized correlated singleton two-sided market has a pure Nash equilibrium 
and it can be computed with a finite number of strategy changes. 
What is more, 
the proof for Theorem \ref{THMpbsingletonIP} applies to 
a generalized correlated singleton two-sided market, 
and hence it is indeed a potential game. 

\begin{theorem}[$\star$]
\label{THMtwosidedsingleton}
    A generalized correlated two-sided singleton market 
    $(N,E, (\mathcal{S}_i), (c_{i,e}), (d_e))$
    with ties 
    is a potential game, and hence possesses a pure Nash equilibrium. 
\end{theorem}

% Theorem \ref{THMtwosidedsingleton} can be proved in the same way as 
% we have proved Theorem \ref{THMpbsingletonIP}, 
% and thus we defer the proof for Theorem \ref{THMtwosidedsingleton} 
% in Appendix \ref{APPproof}. 

\section{Extension Beyond Singleton Games}
\label{SECbeyond}

\full{
In this section, 
we discuss extensions of the above results on priority-based singleton congestion games 
into larger classes with respect to the strategy spaces. 
We present extensions of Theorems \ref{THMpbsingletonCP}, \ref{THMpbsingletonIP}, \ref{THMpspbsingletonIP}, 
and \ref{THMtwosidedsingleton} into matroid games, 
followed by an investigation of priority-based congestion games with consistent priorities 
without any assumption on the strategy spaces of the players. 
}

\full{
\subsection{Matroid Games}
}

\conf{We first discuss matroid games. }%
The following is a fundamental property of matroids, 
which 
is essential to the extension of 
our arguments for singleton games to matroid games. 

\begin{lemma}[see, e.g.,\ \cite{Mur03,Sch03}]
\label{LEMmatroid}
    Let $(E,\mathcal{S})$ be a matroid, 
    $S\in \mathcal{S}$ be a base, 
    and 
    $w_e \in \RR$ be a weight for each $e\in E$. 
    If there exists a base $S'\in \mathcal{S}$ 
    such that $\sum_{e\in S'}w_e < \sum_{e\in S}w_e$, 
    then there exists an element $e\in S$ and $e'\in E\setminus S$ 
    such that 
    $(S \setminus \{e\})\cup \{e'\} \in \mathcal{S}$ 
    and 
    $w_{e'} < w_e$. 
\end{lemma}

It follows from Lemma \ref{LEMmatroid} that 
we can implement an arbitrary better response of a player in a matroid game 
as a lazy better response. 
On the basis of this fact, 
the proofs for Theorems \ref{THMpbsingletonCP}, \ref{THMpbsingletonIP}, \ref{THMpspbsingletonIP}, 
and \ref{THMtwosidedsingleton} can be adapted to matroid games. 

\begin{theorem}
\label{THMpspbmatroidCP}
In a priority-based player-specific matroid congestion game with consistent priorities, 
there exists a sequences of polynomial number of better responses 
starting from an arbitrary strategy profile and 
reaching a pure Nash equilibrium. 
\end{theorem}

\begin{theorem}
\label{THMpbmatroidIP}
A priority-based matroid congestion game with inconsistent priorities 
%%$(N,E, (\mathcal{S}_i), (p_{e}), (d_e))$
is a potential game, 
and hence possesses a pure Nash equilibrium. 
\end{theorem}

\begin{theorem}
\label{THMpspbmatroidIP}
A priority-based player-specific matroid congestion game 
with inconsistent priority 
possesses a pure Nash equilibrium, 
which can be computed with a finite number of strategy changes. 
\end{theorem}

\begin{theorem}
\label{THMtwosidedmatroid}
    A generalized correlated two-sided matroid market with ties 
    %%$(N,E, (\mathcal{S}_i), (c_{i,e}), (d_e))$
    is a potential game, and hence possesses a pure Nash equilibrium. 
\end{theorem}

\full{\subsection{Arbitrary Strategy Spaces}}
\conf{We next deal with arbitrary strategy spaces. }%
For a priority-based congestion game $(N, E, (\mathcal{S}_i)_{i\in N}, p, (d_e)_{e\in E})$  with consistent priorities, 
let $N^ q$ denote the set of the players with priority-function value $ q$, 
namely 
$N^ q=\{i\in N \colon p(i) =  q\}$.

\begin{lemma}[$\star$]
\label{LEMarbit}
% For a priority-based congestion game with consistent proprieties 
% and its strategy profile, 
% a game restricted to the players with a fixed priority is a potential game. 
Let $G=(N, E, (\mathcal{S}_i)_{i\in N}, p, (d_e)_{e\in E})$ 
be a priority-based congestion game with consistent priorities. 
Let $S=(S_1,\ldots, S_n)$ be its strategy profile and 
let $ q \in \ZZ_+$. 
Fix the strategy of each player $j\in N\setminus N^ q$ to $S_j$, 
and let $G^q$ denote the game restricted to the players in $N^ q$. 
Then, 
the game $G^ q$ is a potential game 
with potential function 
\full{\begin{align*}
    \Phi(S^ q) = \sum_{e\in E} \sum_{k=1}^{n_e(S^ q)} d_e(n_e^{< q}(S),k) \quad (S^q=(S_i)_{i\in N^q}). 
\end{align*}}%
\conf{$\Phi(S^ q) = \sum_{e\in E} \sum_{k=1}^{n_e(S^ q)} d_e(n_e^{< q}(S),k) \ (S^q=(S_i)_{i\in N^q})$.}
\end{lemma}

It directly follows from Lemma \ref{LEMarbit} 
that 
a pure Nash equilibrium of a priority-based congestion game $G$ with consistent priorities 
can be constructed 
by combining 
pure Nash equilibria $S^q$ of a game $G^q$ for each priority-function value $q$, 
where $G^q$ is defined by fixing the strategies of the players in $N^{q'}$ 
to form a pure Nash equilibrium of $G^{q'}$ for each $q' < q $.

\begin{theorem}
\label{THMgeneral}
A priority-based congestion game with consistent priorities possesses a pure Nash equilibrium. 
\end{theorem}

\section{Conclusion}

We have presented a common generalization of the models of congestion games by 
Ackermann et al.\ \cite{AGMRV08} 
and Bil\`{o} and Vinci \cite{BV23tcs}. 
This generalization gives a positive and full answer to the open question posed 
by Ackermann et al.\ \cite{AGMRV08}. 
We then proved some theorems on the existence of pure Nash equilibria, 
extending those in \cite{AGMRV08} and \cite{BV23tcs}. 

Once the existence of pure Nash equilibria is established, 
a possible direction of future work is to 
design an efficient algorithm for finding a pure Nash equilibrium in our model. 
Analyses on the 
price of anarchy and the price of stability in our model  
is also of interest, 
as is intensively done for the model of Bil\`{o} and Vinci \cite{BV23tcs}.

%%\clearpage

%%\section*{Acknowledgements}
%%This work is partially supported by 
%%JSPS KAKENHI Grant Number 
%%JP20K11699, 
%%Japan.

% In the interest of anonymization, please do not include acknowledgements in your submission.
%
%\begin{acks}
%
%	The authors would like to thank Dr. Maura Turolla of Telecom
%	Italia for providing specifications about the application scenario.
%
%	The work is supported by the \grantsponsor{GS501100001809}{National
%		Natural Science Foundation of
%		China}{http://dx.doi.org/10.13039/501100001809} under Grant
%	No.:~\grantnum{GS501100001809}{61273304\_a}
%	and~\grantnum[http://www.nnsf.cn/youngscientsts]{GS501100001809}{Young
%		Scientsts' Support Program}.
%
%
%\end{acks}

% Bibliography
%\bibliographystyle{ACM-Reference-Format}
%\bibliography{sample-bibliography}
%\bibliography{refs}

% %%\bibliographystyle{plain}
%\bibliographystyle{plainurl}% the mandatory bibstyle
 \bibliographystyle{abbrv}
 \bibliography{refs}

\begin{thebibliography}{10}

\bibitem{AGMRV08}
H.~Ackermann, P.~W. Goldberg, V.~S. Mirrokni, H.~R{\"{o}}glin, and
  B.~V{\"{o}}cking.
\newblock A unified approach to congestion games and two-sided markets.
\newblock {\em Internet Math.}, 5(4):439--457, 2008.

\bibitem{ARV08}
H.~Ackermann, H.~R{\"{o}}glin, and B.~V{\"{o}}cking.
\newblock On the impact of combinatorial structure on congestion games.
\newblock {\em J. {ACM}}, 55(6):25:1--25:22, 2008.

\bibitem{ARV09}
H.~Ackermann, H.~R{\"{o}}glin, and B.~V{\"{o}}cking.
\newblock Pure {Nash} equilibria in player-specific and weighted congestion
  games.
\newblock {\em Theor. Comput. Sci.}, 410(17):1552--1563, 2009.

\bibitem{BB00}
M.~Ba{\"{\i}}ou and M.~Balinski.
\newblock Many-to-many matching: stable polyandrous polygamy (or polygamous
  polyandry).
\newblock {\em Discret. Appl. Math.}, 101(1-3):1--12, 2000.

\bibitem{BAM07}
V.~Bansal, A.~Agrawal, and V.~S. Malhotra.
\newblock Polynomial time algorithm for an optimal stable assignment with
  multiple partners.
\newblock {\em Theor. Comput. Sci.}, 379(3):317--328, 2007.

\bibitem{BV23tcs}
V.~Bil{\`{o}} and C.~Vinci.
\newblock Congestion games with priority-based scheduling.
\newblock {\em Theor. Comput. Sci.}, 974:114094, 2023.

\bibitem{BV23book}
V.~Bil{\`{o}} and C.~Vinci.
\newblock {\em Coping with Selfishness in Congestion Games---Analysis and
  Design via {LP} Duality}.
\newblock Monographs in Theoretical Computer Science. An {EATCS} Series.
  Springer, 2023.

\bibitem{DFRS19}
M.~Drees, M.~Feldotto, S.~Riechers, and A.~Skopalik.
\newblock Pure {N}ash equilibria in restricted budget games.
\newblock {\em J. Comb. Optim.}, 37(2):620--638, 2019.

\bibitem{EMMM12}
P.~Eirinakis, D.~Magos, I.~Mourtos, and P.~Miliotis.
\newblock Finding all stable pairs and solutions to the many-to-many stable
  matching problem.
\newblock {\em {INFORMS} J. Comput.}, 24(2):245--259, 2012.

\bibitem{FLS17}
M.~Feldotto, L.~Leder, and A.~Skopalik.
\newblock Congestion games with complementarities.
\newblock In D.~Fotakis, A.~Pagourtzis, and V.~T. Paschos, editors, {\em 10th
  International Conference on Algorithms and Complexity, {CIAC} 2017}, volume
  10236 of {\em Lecture Notes in Computer Science}, pages 222--233, 2017.

\bibitem{FLS18}
M.~Feldotto, L.~Leder, and A.~Skopalik.
\newblock Congestion games with mixed objectives.
\newblock {\em J. Comb. Optim.}, 36(4):1145--1167, 2018.

\bibitem{Fle03}
T.~Fleiner.
\newblock On the stable $b$-matching polytope.
\newblock {\em Math. Soc. Sci.}, 46(2):149--158, 2003.

\bibitem{FGHPZ15}
S.~Fujishige, M.~X. Goemans, T.~Harks, B.~Peis, and R.~Zenklusen.
\newblock Congestion games viewed from {M}-convexity.
\newblock {\em Oper. Res. Lett.}, 43(3):329--333, 2015.

\bibitem{GS62}
D.~Gale and L.~Shapley.
\newblock College admissions and the stability of marriage.
\newblock {\em Am. Math. Mon.}, 69:9--15, 1962.

\bibitem{GMMY22}
H.~Goko, K.~Makino, S.~Miyazaki, and Y.~Yokoi.
\newblock Maximally satisfying lower quotas in the hospitals/residents problem
  with ties.
\newblock In P.~Berenbrink and B.~Monmege, editors, {\em 39th International
  Symposium on Theoretical Aspects of Computer Science, {STACS} 2022}, volume
  219 of {\em LIPIcs}, pages 31:1--31:20, 2022.

\bibitem{GI89}
D.~Gusfield and R.~W. Irving.
\newblock {\em The Stable Marriage Problem---Structure and Algorithms}.
\newblock Foundations of computing series. {MIT} Press, 1989.

\bibitem{HMY19}
K.~Hamada, S.~Miyazaki, and H.~Yanagisawa.
\newblock Strategy-proof approximation algorithms for the stable marriage
  problem with ties and incomplete lists.
\newblock In P.~Lu and G.~Zhang, editors, {\em 30th International Symposium on
  Algorithms and Computation, {ISAAC} 2019}, volume 149 of {\em LIPIcs}, pages
  9:1--9:14, 2019.

\bibitem{HKP18}
T.~Harks, M.~Klimm, and B.~Peis.
\newblock Sensitivity analysis for convex separable optimization over integral
  polymatroids.
\newblock {\em {SIAM} J. Optim.}, 28(3):2222--2245, 2018.

\bibitem{HP15}
T.~Harks and B.~Peis.
\newblock Resource buying games.
\newblock In A.~S. Schulz, M.~Skutella, S.~Stiller, and D.~Wagner, editors,
  {\em Gems of Combinatorial Optimization and Graph Algorithms}, pages
  103--111. Springer, 2015.

\bibitem{HT18}
T.~Harks and V.~Timmermans.
\newblock Uniqueness of equilibria in atomic splittable polymatroid congestion
  games.
\newblock {\em J. Comb. Optim.}, 36(3):812--830, 2018.

\bibitem{HIMY15}
C.-C. Huang, K.~Iwama, S.~Miyazaki, and H.~Yanagisawa.
\newblock A tight approximation bound for the stable marriage problem with
  restricted ties.
\newblock In N.~Garg, K.~Jansen, A.~Rao, and J.~D.~P. Rolim, editors, {\em
  Approximation, Randomization, and Combinatorial Optimization. Algorithms and
  Techniques, {APPROX/RANDOM} 2015}, volume~40 of {\em LIPIcs}, pages 361--380,
  2015.

\bibitem{HK15}
C.-C. Huang and T.~Kavitha.
\newblock Improved approximation algorithms for two variants of the stable
  marriage problem with ties.
\newblock {\em Math. Program.}, 154(1-2):353--380, 2015.

\bibitem{IMNSS05}
S.~Ieong, R.~McGrew, E.~Nudelman, Y.~Shoham, and Q.~Sun.
\newblock Fast and compact: {A} simple class of congestion games.
\newblock In M.~M. Veloso and S.~Kambhampati, editors, {\em 20th Annual AAAI
  Conference on Artificial Intelligence, AAAI 2005}, pages 489--494, 2005.

\bibitem{Irv94}
R.~W. Irving.
\newblock Stable marriage and indifference.
\newblock {\em Discret. Appl. Math.}, 48(3):261--272, 1994.

\bibitem{Kam15}
N.~Kamiyama.
\newblock Stable matchings with ties, master preference lists, and matroid
  constraints.
\newblock In M.~Hoefer, editor, {\em 8th International Symposium, {SAGT} 2015},
  volume 9347 of {\em Lecture Notes in Computer Science}, pages 3--14.
  Springer, 2015.

\bibitem{Kam19}
N.~Kamiyama.
\newblock Many-to-many stable matchings with ties, master preference lists, and
  matroid constraints.
\newblock In E.~Elkind, M.~Veloso, N.~Agmon, and M.~E. Taylor, editors, {\em
  18th International Conference on Autonomous Agents and MultiAgent Systems,
  {AAMAS} 2019}, pages 583--591. IFAAMAS, 2019.

\bibitem{Kav22}
T.~Kavitha.
\newblock Stable matchings with one-sided ties and approximate popularity.
\newblock In A.~Dawar and V.~Guruswami, editors, {\em 42nd {IARCS} Annual
  Conference on Foundations of Software Technology and Theoretical Computer
  Science, {FSTTCS} 2022}, volume 250 of {\em LIPIcs}, pages 22:1--22:17, 2022.

\bibitem{KT22}
F.~Kiyosue and K.~Takazawa.
\newblock A common generalization of budget games and congestion games.
\newblock In P.~Kanellopoulos, M.~Kyropoulou, and A.~A. Voudouris, editors,
  {\em 15th International Symposium on Algorithmic Game Theory, {SAGT} 2022},
  volume 13584 of {\em Lecture Notes in Computer Science}, pages 258--274.
  Springer, 2022.

\bibitem{Man13}
D.~F. Manlove.
\newblock {\em Algorithmics of Matching Under Preferences}, volume~2 of {\em
  Series on Theoretical Computer Science}.
\newblock WorldScientific, 2013.

\bibitem{MS10}
D.~Marx and I.~Schlotter.
\newblock Parameterized complexity and local search approaches for the stable
  marriage problem with ties.
\newblock {\em Algorithmica}, 58(1):170--187, 2010.

\bibitem{Mil96}
I.~Milchtaich.
\newblock Congestion games with player-specific payoff functions.
\newblock {\em Game. Econ. Behav.}, 13:111--124, 1996.

\bibitem{MS96}
D.~Monderer and L.~S. Shapley.
\newblock Potential games.
\newblock {\em Games Econ. Behav.}, 14:124--143, 1996.

\bibitem{Mur03}
K.~Murota.
\newblock {\em Discrete Convex Analysis}.
\newblock Society for Industrial and Applied Mathematics, Philadelphia, 2003.

\bibitem{NRTV07}
N.~Nisan, T.~Roughgarden, {\'E}.~Tardos, and V.~V. Vazirani, editors.
\newblock {\em Algorithmic Game Theory}.
\newblock Cambridge University Press, 2007.

\bibitem{Ros73a}
R.~W. Rosenthal.
\newblock A class of games possessing pure-strategy {Nash} equilibria.
\newblock {\em Int. J. Game Theory}, 2:65--67, 1973.

\bibitem{Rou16}
T.~Roughgarden.
\newblock {\em Twenty Lectures on Algorithmic Game Theory}.
\newblock Cambridge University Press, 2016.

\bibitem{Sch03}
A.~Schrijver.
\newblock {\em Combinatorial Optimization---Polyhedra and Efficiency}.
\newblock Springer, 2003.

\bibitem{Sot99a}
M.~Sotomayor.
\newblock The lattice structure of the set of stable outcomes of the multiple
  partners assignment game.
\newblock {\em Int. J. Game Theory}, 28(4):567--583, 1999.

\bibitem{Sot99b}
M.~Sotomayor.
\newblock Three remarks on the many-to-many stable matching prblem.
\newblock {\em Math. Soc. Sci.}, 38(1):55--70, 1999.

\bibitem{Tak19}
K.~Takazawa.
\newblock Generalizations of weighted matroid congestion games: pure {N}ash
  equilibrium, sensitivity analysis, and discrete convex function.
\newblock {\em J. Comb. Optim.}, 38(4):1043--1065, 2019.

\bibitem{Tak24}
K.~Takazawa.
\newblock Pure {N}ash equilibria in weighted congestion games with
  complementarities and beyond.
\newblock In M.~Dastani, J.~S. Sichman, N.~Alechina, and V.~Dignum, editors,
  {\em Proceedings of the 23rd International Conference on Autonomous Agents
  and Multiagent Systems, {AAMAS} 2024}, pages 2495--2497. {ACM}, 2024.

\bibitem{Yok21}
Y.~Yokoi.
\newblock An approximation algorithm for maximum stable matching with ties and
  constraints.
\newblock In H.~Ahn and K.~Sadakane, editors, {\em 32nd International Symposium
  on Algorithms and Computation, {ISAAC} 2021}, volume 212 of {\em LIPIcs},
  pages 71:1--71:16, 2021.

\end{thebibliography}

% Appendix
\appendix

\section{Omitted Proof}
\label{APPproof}

\subsection{Proofs from Section \ref{SECsingleton}}

\full{In a singleton game, 
a strategy $\{e\}$ is simply denoted by $e$. 
}%
We use a term \emph{state} to refer to a collection of the strategies of some of the players in $N$,  
and 
let $N(S)$ denote the set of the players contributing a state $S$. 
In other words, 
a strategy profile is a special case of a state $S$ for which $N(S)=N$. 
For a state $S$ and a resource $e\in E$, 
let $N_e(S)$ denote the set of players choosing $e$ as the strategy, 
and let $|N_e(S)| = n_e(S)$.

\begin{proof}[Proof of Theorem \ref{THMpbsingletonCP}]
Let $\{ q_1, q_2,\ldots,  q_{k}\}$ denote 
the set of the priority-function values of all players, 
i.e.,\
\begin{align*}
    \mbox{$\{ q_1, q_2,\ldots,  q_{k}\}=\{p(i)\colon i\in N\}$, where $q_1 < q_2 <\cdots <  q_{k}$.}
\end{align*}
For each $k' =1,2,\ldots,{k}$, 
define  
$N^{k'} \subseteq N$ by 
$N^{k'} = \{i\in N \colon p(i) =  q_{k'}\}$. 

Let $S=(e_1,\ldots, e_n)$ be an arbitrary strategy profile of $G$, 
and 
let $S^{k'}$ be a state of $G$ consisting of the strategies of the players in $N^{k'}$ in $S$ 
for each ${k'} =1,2,\ldots, {k}$. 
We prove the theorem by induction on $k'$. 
First, 
define a player-specific singleton congestion game 
$$G^1= (N^1,E,(\mathcal{S}_i)_{i\in N^1}, (d_{i,e}')_{i\in N^1,e\in E} )$$ 
in which 
the delay function $d'_{i,e}\colon \ZZ_{++}\to \RR_+$ ($i\in N^1$, $e\in E$) is defined by 
\begin{align*}
    d_{i,e}' (y) = d_{i,e}(0  ,y) \quad (y\in \ZZ_{++}). 
\end{align*}
It then follows from Theorem \ref{THMpssingleton} that 
$G^1$ has a pure Nash equilibrium $\hat{S}^1$, 
which is attained by a polynomial number of best responses from $S^1$. 

Now 
let $k'\in \{1,\ldots, k-1\}$ and 
suppose that we have a state $\hat{S}^{k'}$ of the players in $\bigcup_{\ell=1}^{k'} N^\ell$ 
in which no player has an incentive to change her strategy. 
Then 
construct a player-specific singleton congestion game 
$$G^{{k'}+1}=(N^{{k'}+1},E,(\mathcal{S}_i)_{i\in N^{{k'}+1}}, (d'_{i,e})_{i\in N^{k'+1},e\in E})$$ 
in which 
\begin{align*}
d_{i,e}' (y) = d_{i,e}(n_e(\hat{S}^{k'})  ,y) \quad (y\in \ZZ_{++})
\end{align*}
for each $e\in E$. 
It again follows from Theorem \ref{THMpssingleton} that 
the game $G^{{k'}+1}$ has a pure Nash equilibrium and 
it is attained by a polynomial number of best responses from 
an arbitrary strategy profile. 

By induction, 
we have proved that a pure Nash equilibrium of a player-specific priority-based singleton congestion game 
can be attained through a polynomial number of best responses from an arbitrary strategy profile. 
\end{proof}

\begin{proof}[Proof of Theorem \ref{THMpspbsingletonIP}]
We prove this theorem by presenting an algorithm for computing a pure Nash equilibrium 
of a priority-based player-specific singleton congestion game 
$
(N, E, (\mathcal{S}_i), (p_e), (d_{i,e}))
$. 
The algorithm constructs a sequence $S_0,S_1,\ldots, S_k$ of states 
in which $N(S_0)=\emptyset$, $N(S_k)=N$, 
and 
\begin{align}
\label{EQnoincentive}
\mbox{each player in $N(S_{k'})$ has no incentive to change her strategy}
\end{align}
for each ${k'}=0,1,\ldots, k$,  
implying that $S_k$ is a pure Nash equilibrium. 

It is clear that \eqref{EQnoincentive} is satisfied for ${k'} =0$. 
Below we show how to construct $S_{{k'}+1}$ from $S_{k'}$ 
under an assumption that $S_{k'}$ satisfies \eqref{EQnoincentive} and $N(S_{k'})\subsetneq N$.

Take a player $i\in N \setminus N(S_{k'})$, 
and let $i$ choose a resource $e\in E$ 
imposing the minimum cost on $i$ 
if $i$ is added to $N_e(S_{k'})$. 
We construct the new state $S_{{k'}+1}$ by changing 
the strategy of each player $j \in N_e(S_{k'})$ in the following way. 
The other players do not change their strategies. 

For the players in $N_e(S_{k'})$, 
we have the following cases A and B. 
\begin{description}
\item[Case A.]
    No player in $N_e(S_{k'})$ comes to have 
    a better response 
    when $i$ is added to $N_e(S_{k'})$. 
\item[Case B.]
        Some players in $N_e(S_{k'})$ comes to have a better response 
        %%due to the involvement of $i$.     
        when $i$ is added to $N_e(S_{k'})$. 
\end{description}

In Case A, 
we do not change the strategies of the players in $N_e(S_{k'})$.  
In Case B, 
if a player $j\in N_e(S_{k'})$ 
comes to have a better response, 
it must hold that $p_e(j)\ge p_e(i)$. 
We further separate Case B into the following two cases. 
\begin{description}
\item[Case B1.]
    There exists a player $j\in N_e(S_{k'})$ having a better response 
    and satisfying $p_e(j) = p_e(i)$. 
\item[Case B2.]
    Every player $j\in N_e(S_{k'})$ having a better response 
    satisfies
    $p_e(j) > p_e(i)$. 
\end{description}
In each case, 
the strategies are changed as follows. 
\begin{description}
    \item[Case B1.] 
        Only one player $j\in N_e(S_{k'})$ having a better response and $p_e(j)=p_e(i)$ 
        changes her strategy by 
        discarding her strategy. 
        Namely, $j \not \in N(S_{k'+1})$. 
        The other players do not change their strategies. 
        % even if they satisfy the same condition as $j$, 
        % i.e.,\ having a better response and the same priority as $i$. 
    \item[Case B2.]
        Every player $j\in N_e(S_{k'})$ having a better response discards her strategy. 
\end{description}

We have now constructed the new state $S_{k'+1}$. 
It is straightforward to see that the state $S_{k'+1}$ satisfies \eqref{EQnoincentive}. 
We complete the proof by showing that this algorithm terminates within a finite number of strategy changes.

For a resource $e\in E$,  
let $ q^*_e = \max\{ p_e(i)\colon i\in N \}$. 
For each state $S_{k'}$ appearing in the algorithm, 
define its potential 
$\Phi(S_{k'}) \in \left(\bigtimes_{e\in E} \ZZ_+^{ q_e^*}\right) \times \ZZ_{++}$
in the following manner. 
For each resource $e\in E$, 
define a 
vector $\phi_e \in \ZZ_+^{ q_e^*}$ 
by 
\begin{align*}
    \phi_e( q) = n_e^ q(S_{k'}) \quad ( q = 1,2,\ldots,  q_e^*), 
\end{align*}
which is a contribution of $e$ to the first component of $\Phi(S_{k'})$. 
The first component of $\Phi(S_{k'})$ is constructed by 
ordering the vectors $\phi_e$ ($e\in E$) in the lexicographically nondecreasing order.

% The second component of $\Phi(S_{k'})$ is a positive integer 
% defined as follows. 
For a resource $e\in E$ and a player $i\in N_e(S_{k'})$, 
define $\tol(i,S_{k'}) \in \ZZ_{++}$ as 
the maximum number $y\in \ZZ_{++}$ such that 
$e$ is an optimal strategy for $i$ if 
$i$ shares $e$ with $y$ players having the same priority as $e$, 
i.e.,\ 
\begin{align*}
        d_{i,e}( n_e^{<p_e(i)}(S_{k'}), y ) \le 
        d_{i,e'}( n_{e'}^{<p_{e'}(i)}(S_{k'}), n_{e'}^{p_{e'}(i)}(S_{k'})+1 ) 
\end{align*}
        for each $e'$ with $e'\neq e$ and $\{e'\}\in \mathcal{S}_i$. 
        Note that $i$ herself is counted in $y$, 
and hence $\tol(i,S_{k'})\ge 1$ for each $i\in N(S_{k'})$. 
Now the second component of the potential $\Phi(S_{k'})$ is defined as 
%\begin{align*}
$\sum_{i\in N(S_{k'})} \tol(i,S_{k'})$. 
%\end{align*}

We prove that the potential $\Phi(S_{k'})$ increases lexicographically monotonically during the algorithm. 
Let a state $S_{k'+1}$ is constructed from $S_{k'}$ and the involvement of  
a player $i\in N \setminus N(S_{k'})$ 
choosing a resource $e\in E$. 
It is straightforward to see that 
% \begin{align*}
%     \phi_{e'}(S_{\ell+1}) = \phi_{e'}(S_{\ell}) \quad (e' \in E \setminus \{e\}). 
% \end{align*}
%%$\phi_{e'}(S_{k'})$ is unchanged for each $e' \in E \setminus \{e\}$. 
$\phi_{e'}$ is unchanged for each $e' \in E \setminus \{e\}$. 
%%Consider the vectors $\phi_{e}(S_{k'})$ and $\phi_{e}(S_{{k'}+1})$. 
Consider how the vector $\phi_{e}$ changes. 

%%\paragraph{Case A.}
\subparagraph*{Case A.}
%%The unique difference between $\phi_{e}(S_{k'+1})$ and $\phi_{e}(S_{k'})$ is that 
The unique change of $\phi_{e}$ is that 
$\phi_e(p_e(i))$ increases by one, 
implying that the first component of $\Phi(S_{k'+1})$ is lexicographically larger than 
that of $\Phi(S_{k'})$.

%%\paragraph{Case B1.}
\subparagraph*{Case B1.}
Let $j^*\in N_e(S_{k'})$ denote the unique player who discard her strategy. 
Recall that $p_e(j^*) = p_e(i)$. 
% It follows that $\phi_{e'}$ is unchanged for each $e'\in E$, 
% including $e$, 
It follows that $\phi_{e}$ is unchanged, 
and hence the first component of $\Phi(S_{k'+1})$ is the same as that of $\Phi(S_{k'})$. 
The second component of $\Phi(S_{{k'}+1})$ is strictly larger than that of $\Phi(S_{k'})$, 
because 
\begin{align*}
    {}&{}N(S_{{k'}+1}) = (N(S_{k'})\cup \{i\})\setminus\{j^*\}, \\
    {}&{}\tol(j,S_{{k'}+1}) = \tol(j,S_{{k'}}) \quad \mbox{for each $j\in N(S_{k'}) \setminus \{i,j^*\}$}, \\
    {}&{}\tol(i,S_{{k'}+1}) \ge n_e^{p_e(i)}(S_{k'})+1, \\
    {}&{}\tol(j^*,S_{{k'}}) = n_e^{p_e(i)}(S_{k'}).
\end{align*}

%%\paragraph{Case B2.} 
\subparagraph*{Case B2.} 
    % For each resource $e' \in E \setminus \{e\}$, 
    % it is clear that 
    % $\phi_{e'}(S_{{k'}+1}) = \phi_{e'}(S_{k'})$. 
    %%For the resource $e$, 
    It holds that 
    % \begin{align*}
    %     {}&{}n_e^ q(S_{{k'}+1}) = n_e^ q(S_{k'}) \quad \mbox{for each $q < p_e(i)$}, \\
    %     {}&{}n_e^{p_e(i)}(S_{{k'}+1}) = n_e^{p_e(i)}(S_{k'})+1. 
    % \end{align*}
    $n_e^ q(S_{{k'}+1}) = n_e^ q(S_{k'})$ for each $q < p_e(i)$ 
    and 
    $n_e^{p_e(i)}(S_{{k'}+1}) = n_e^{p_e(i)}(S_{k'})+1$. 
    % %%while 
    %%$n_e^ q(S_{{k'}+1}) \le n_e^ q(S_{{k'}})$ for $ q > p_e(i)$. 
    Thus, 
    the first component of $\Phi$  lexicographically increases. 
\end{proof}

\subsection{Proofs from Section \ref{SECg-corr}}

\begin{proof}[Proof of Proposition \ref{PROPnonps}]
Given a priority-based congestion game $(N,E,(\mathcal{S}_i),(p_e), (d_e) )$ with inconsistent priorities, 
construct a generalized correlated two-sided market $(N,E,(\mathcal{S}_i),(c_{i,e}), (d_e') )$ with ties 
by defining 
\begin{alignat*}{2}
    {}&{}c_{i,e} = p_e(i) &         \quad &(i\in N, e \in E), \\
    {}&{}d_e'(c,x,y) = d_e(x,y).&   \quad &(e\in E, c\in \RR_+, x\in \ZZ_+, y\in \ZZ_{++}). 
\end{alignat*}
It is straightforward to see that 
the delay function $d_e'$ ($e\in E$) satisfy
\eqref{EQmonotonec3}--\eqref{EQplusone3} in which $d_e$ is replaced by $d_e'$, 
if the original delay function $d_e$ satisfies \eqref{EQmonotonex}--\eqref{EQplusone}. 
\end{proof}

\begin{proof}[Proof of Proposition \ref{PROPps}]
Given 
a generalized correlated two-sided market $(N,E,(\mathcal{S}_i),(c_{i,e}), (d_e) )$ with ties, 
construct 
a priority-based player-specific congestion game $(N,E,(\mathcal{S}_i),(p_e), (d_{i,e}') )$ with inconsistent priorities 
as follows. 
For each resource $e\in E$, 
construct its priority function $p_e\colon N \to \ZZ_+$ 
in the same way as in \eqref{EQpriority}, 
and 
define its delay function $d_{i,e}'\colon \ZZ_+\times\ZZ_{++}\to \RR_+$ specific to a player $i\in N$ by 
\begin{alignat*}{2}
    {}&{}d_{i,e}'(x,y) = d_e(p_e(i),x,y) \quad (x\in \ZZ_+, y\in \ZZ_{++}). 
\end{alignat*}
It is straightforward to see that 
the delay function $d_{i,e}'$ ($i\in N$, $e\in E$) satisfies 
\eqref{EQmonotonex}--\eqref{EQplusone} 
in which $d_e$ is replaced by $d_{i,e}'$, 
if the original delay function $d_e$ satisfies \eqref{EQmonotonec3}--\eqref{EQplusone3}.  
\end{proof}

\begin{proof}[Proof of Theorem \ref{THMtwosidedsingleton}]
For each strategy profile $S=(e_1,\ldots, e_n)$, 
define its potential $\Phi(S) \in (\RR_+ \times \ZZ_{++})^n$ as follows. 
Let $e\in E$ be a resource, 
and 
let $Q_e(S)=\{ q_1,..., q_{k^*}\}$ be a set of integers such that 
$Q_e(S)=\{q\colon n_e^q(S) > 0\}$ and $ q_1 < \cdots <  q_{k^*}$. 
% While $q_1,\ldots, q_{k^*}$ depend on $S$ and $e$, 
% for conciseness we omit $S$ and $e$ in this notation. 
The resource $e\in E$ contributes the following 
$n_e(S)$ vectors in $\RR_+ \times \ZZ_{++}$
to $\Phi(S)$:
\begin{align*}
&{}(d_e( q_1,0,1),\,  q_1 ), \ldots, (d_e( q_1,0,n_e(S, q_1)),\,  q_1), \\
&{}(d_e( q_2, n_e^{q_1}(S),1),\,  q_2), \ldots, 
(d_e ( q_2,n_e^{q_1}(S),n_e^{q_2}(S) ),\,  q_2 ), \\
&{}\ldots, \\
&{}(d_e( q_k, n_e^{<q_k}(S), 1 ),\,  q_k), \ldots, 
( d_e( q_k,n_e^{<q_k}(S), n_e^{q_k}(S)),\,  q_k), \\
&{}\ldots, \\
&{}(d_e( q_{k^*},n_e^{<q_{k^*}}(S), 1 ),\,  q_{k^*}), \ldots, 
( d_e( q_{k^*},n_e^{<q_{k^*}}(S), n_e^{q_{k^*}}(S)),\,  q_{k^*}). 
\end{align*}

The potential $\Phi(S)$ is obtained by 
ordering the $n$ vectors contributed by all resources in the lexicographically nondecreasing order. 
We can observe that the order of the $n_e(S)$ vectors shown above is lexicographically nondecreasing 
in the following way. 
It follows from the property \eqref{EQmonotoney3} that 
\begin{align*}
{}&{}d_e( q_k, n_e^{<q_{k}}(S), y ) \le 
d_e( q_k,n_e^{<q_{k}}(S), y+1 ) 
%%\le \cdots \\
%%{}&{}
%%\le
%%d_e\left( q_k,n_e(S,  q_{k}^<), n_e(S,  q_k)\right) 
\end{align*}
for each $k=1,\ldots, k^*$ and 
for each $y = q_1,\ldots, q_{k^ -1}$. 
It further follows 
from the properties \eqref{EQmonotonec3}, \eqref{EQmonotonex3} and \eqref{EQplusone3} that 
\begin{align}
\label{EQnextcxy}
d_e (c,x, y )
\le d_e (c, x+y-1,1 ) 
%%\le d_e\left( q_k,n_e(S,  q_{k+1}^<),1\right) 
\le d_e (c', x+y,1 ) 
\end{align}
if $c < c'$, 
and in particular 
\begin{align}
\label{EQnext3}
d_e( q_k,n_e^{<q_k}(S), n_e^{q_k}(S))
\le d_e( q_{k+1}, n_e^{<q_{k+1}}(S),1)
\quad 
\mbox{for each $k = 1,\ldots, k^* -1$.}
\end{align}

Suppose that a player $i$ has a better response 
in a strategy profile $S$, 
which changes her strategy from $e$ to $e'$, 
and let $S'=(S_{-i}, e')$. 
Below we show that $\Phi(S') \preclex \Phi(S)$, 
which completes the proof. 

Let $c_{i,e} =  q$ 
and 
$c_{i,e'} =  q'$. 
Since the delay imposed on $i$ becomes smaller due to the better response, 
it holds that 
\begin{align}
\label{EQbr3}
d_{e'}(
%    c_{i,e'}, n_{e'}^{<c_{i,e'}}(S), n_{e'}^{c_{i,e'}}(S) +1 
    q', n_{e'}^{<q'}(S), n_{e'}^{q'}(S) +1 
)
< 
d_e(
 q,
n_e^{<q}(S), n_e^{q}(S)
). 
\end{align}
Note that 
$e'$ contributes 
a vector 
\begin{align}
\label{EQnewvector3}
( 
d_{e'}(
    q', n_{e'}^{<q'}(S), n_{e'}^{q'}(S) +1 
)
, 
q' ) 
\end{align}
to $\Phi(S')$ but not to $\Phi(S)$. 
To prove $\Phi(S') \preclex \Phi(S)$, 
it suffices to show that 
a vector belonging to $\Phi(S)$ but not to $\Phi(S')$ is 
lexcographically larger than the vector \eqref{EQnewvector3}. 

First, 
consider the vectors in $\Phi(S)$ contributed by $e$. 
Let $Q_e(S)=\{q_1,\ldots, q_{k^*}\}$, 
where $q_1 < \cdots < q_{k^*}$. 
Due to the better response of $i$, 
the vectors in $\Phi(S)$ whose second component is larger than $ q$ changes, 
because the second argument of the delay function $d_e$ decreases by one. 
If 
$q=q_{k^*}$, 
then 
those vectors do not exist 
and thus we are done. 
Suppose that 
$q=q_k$ for some $k<k^*$. 
% those vector exist, 
% i.e.,\ 
% $c_{i,e} <  q_{k^*}$. 
Among those vectors, 
the lexicographically smallest one is 
$$
( 
    d_e(
         q_{k+1}, n_e^{<q_{k+1}}(S), 1
    ),
     q_{k+1}
). 
$$
It follows from 
%%the properties \eqref{EQmonotonex3} and \eqref{EQplusone3} that 
%%$$
%%d_e\left(
%% q_k,n_e(S, q_k^<), n_e(S,  q_k)
%%\right) 
%%\le 
%%d_e\left(
%% q_k,n_e(S, q_{k+1}^{<}) -1,1
%%\right) 
%%\le 
%%d_e\left(
%% q_k,n_e(S, q_{k+1}^{<}),1
%%\right) ,
%%$$
%%and thus 
\eqref{EQnext3} and \eqref{EQbr3} that 
$$
d_{e'}(
    q', n_{e'}^{<q'}(S), n_{e'}^{q'}(S) +1 
)
< 
d_e(
 q,
n_e^{<q}(S), n_e^{q}(S)
)
% =
% d_e\left(
%  q_k,
% n_e^{<q_k}(S), n_e^{q_k}(S)
% \right)
\le 
d_e(
 q_{k+1},n_e^{<q_{k+1}}(S),1
) . 
$$
Hence, 
it holds that 
$$
( 
d_{e'}(
    q', n_{e'}^{<q'}(S), n_{e'}^{q'}(S) +1 
)
, 
q' ) 
%%\left( 
%%d_{e'}\left(
%%    \sum_{q\colon q<p_{e'}(i)} n_{e'}(S,q), n_{e'}(S, p_{e'}(i))+1
%%%%    \sum_{\ell\colon q_{\ell}<p_{e'}(i)}n_{e'}(S, q_{\ell}), n_{e'}(S, p_{e'}(i)+1) 
%%\right),\, p_{e'}(i) \right) 
\preclex
%%\left( 
%%    d_e\left(
%%        \sum_{\ell=1}^{k}n_e(S,q_\ell), 1
%%    \right),\, 
%%    q_{k+1}
%%\right). 
( 
d_e(
 q_{k+1},n_e^{<q_{k+1}}(S),1
),
     q_{k+1}
). 
$$

Next, 
consider the vectors in $\Phi(S)$ contributed by $e'$. 
% Due to the better response of $i$, 
% the vectors whose second component is larger than $c_{i,e'}$ changes. 
Without loss of generality, 
suppose that 
there exists a positive integer $q''$ such that 
$q''\in Q_{e'}(S)$ and $q'' > q'$. 
Let $q''$ be the smallest integer satisfying these conditions. 
% such vectors exist in $\Phi(S)$. 
% i.e.,\ 
% there exists a real number $c' \in \RR_+$ such that 
% $n_{e'}(S, c') > 0$ and $c' > c_{i,e'}$. 
% Among those vectors, 
% the lexicographically smallest one is 
The lexicographically smallest vector in $\Phi(S)$ contributed by $e'$ and changed by the better response of $i$ is 
$$
(
    d_{e'}(
        q'', n_{e'}^{<q''}(S) , 1 
    ),
    q''
). 
$$
% where $c'$ is the smallest real number such that 
% $n_{e'}^{c'}(S) > 0$ and $c'>c$. 
It then follows from 
the properties \eqref{EQmonotonec3} and \eqref{EQplusone3} 
%%\eqref{EQnextcxy} 
that 
$$
    d_{e'}(
        q', n_{e'}^{<q'}(S) , n_{e'}^{q'}(S)+1 
    )
\le 
    d_{e'}(
        q'', n_{e'}^{<q'}(S) , n_{e'}^{q'}(S)+1 
    )
\le
d_{e'}(
    q'', n_{e'}^{< q''}(S) , 1 
)
$$
and thus 
$$
( 
    d_{e'}( 
        q', n_{e'}^{<q'}(S), n_{e'}^{q'}(S)+1 
    ), 
q' ) 
\preclex 
(
    d_{e'}(
        q'', n_{e'}^{< q''}(S) , 1 
    ),
    q''
), 
$$
completing the proof. 
\end{proof}

\subsection{Proof from Section \ref{SECbeyond}}

\begin{proof}[Proof of Lemma \ref{LEMarbit}]
    Let $i$ be a player in $N^ q$, 
    and $S_i\in \mathcal{S}_i$ be an arbitrary strategy of $i$. 
    For each $S_i'\in \mathcal{S}_i$, 
    it holds that 
    \begin{align*}
        \Phi(S^ q_{-i},S_i') - \Phi(S^ q)
        {}&{}= \sum_{e\in S_i' \setminus S_i} d_e(n_e^{< q}(S), n_e(S^ q)+1) - 
        \sum_{e\in S_i\setminus S_i'}d_e(n_e^{< q}(S), n_e(S^ q)) \\
        {}&{} = \gamma_i(S^ q_{-i},S_i') - \gamma_i(S^ q), 
    \end{align*}
    and hence the function $\Phi$ is a potential function of $G^ q$. 
\end{proof}

\end{document}